\documentclass[11pt]{article}
\usepackage[margin=1in]{geometry}
\usepackage[utf8]{inputenc}

\usepackage{booktabs} 

\usepackage[utf8]{inputenc}
\usepackage{authblk}
\usepackage{amsmath, amssymb, amsthm, thmtools, amsfonts, bm, bbm, thm-restate}
\usepackage{algorithm}
\usepackage{algorithmicx}
\usepackage[noend]{algpseudocode}
\usepackage{cite}
\usepackage[numbers,sort]{natbib}

\usepackage{asymptote}
\usepackage{graphicx}
\usepackage{todonotes}
\usepackage{multirow}
\usepackage{comment}
\usepackage{dsfont}
\usepackage{bigstrut}
\usepackage{caption}
\usepackage{subcaption}
\usepackage{multirow}
\usepackage{makecell}
\usepackage{booktabs}
\usepackage{xcolor}
\definecolor{BrickRed}{rgb}{0.8,0.25,0.33}
\usepackage[pagebackref,colorlinks,citecolor=blue,linkcolor=BrickRed]{hyperref}

\usepackage[framemethod=tikz]{mdframed}
\usepackage{nicefrac}

\usepackage{pifont}

\usepackage{cleveref}

\theoremstyle{plain}
\newtheorem{thm}{Theorem}[section]
\newtheorem{cor}[thm]{Corollary}
\newtheorem{prop}[thm]{Proposition}

\newtheorem{lem}[thm]{Lemma}
\newtheorem{lemma}[thm]{Lemma}

\newtheorem{claim}[thm]{Claim}
\newtheorem{Def}[thm]{Definition}
\newtheorem{obs}[thm]{Observation}

\newtheorem{rem}[thm]{Remark}
\newtheorem{example}[thm]{Example}

\crefname{thm}{Theorem}{theorems}
\crefname{cla}{Claim}{claims}
\crefname{lem}{Lemma}{lemmas}
\crefname{fact}{Fact}{facts}

\allowdisplaybreaks

\newcommand{\E}{\mathbb{E}}
\newcommand{\R}{\mathbb{R}}
\newcommand{\bbP}{\Pr}

\newcommand{\eps}{\varepsilon}
\newcommand{\calA}{\mathcal{A}}
\newcommand{\calB}{\mathcal{B}}
\newcommand{\calD}{\mathcal{D}}
\newcommand{\calF}{\mathcal{F}}
\newcommand{\calG}{\mathcal{G}}
\newcommand{\calI}{\mathcal{I}}
\newcommand{\calM}{\mathcal{M}}
\newcommand{\calP}{\mathcal{P}}
\newcommand{\calU}{\mathcal{U}}
\newcommand{\calY}{\mathcal{Y}}

\newcommand{\optoff}{\mathrm{OPT}_\mathrm{off}}
\newcommand{\opton}{\mathrm{OPT}_\mathrm{on}}
\newcommand{\Exp}{\mathrm{Exp}}
\newcommand{\Ber}{\mathrm{Ber}}

\newcommand{\alg}{\texttt{ALG}}

\newcommand*{\todos}{}%

\ifdefined\notodos
\def\neel#1{}
\def\david#1{}
\fi

\ifdefined\todos

\def\neel#1{\marginpar{$\leftarrow$\fbox{N}}\footnote{$\Rightarrow$~{\sf\textcolor{purple}{#1 --Neel}}}}
\def\david#1{\marginpar{$\leftarrow$\fbox{D}}\footnote{$\Rightarrow$~{\sf\textcolor{blue}{#1 --David}}}}
\fi

\let\vec\mathbf
\renewcommand{\vec}{\mathbf}

\title{Combinatorial Stationary Prophet Inequalities}

\author[1]{Neel Patel}
\author[2]{David Wajc\thanks{Part of this work done while the author was at Stanford University and at Google Research. Supported by a Taub Family Foundation ``Leader in Science and Technology'' fellowship.}}
\affil[1]{University of Southern California}
\affil[2]{Technion --- Israel institute of Technology}
\date{\vspace{-1.25cm}}
\begin{document}

\maketitle

\begin{abstract}
    Numerous recent papers have studied the tension between thickening and clearing a market in (uncertain, online) long-time horizon Markovian settings.
    In particular, (Aouad and Sarita{\c{c}} EC'20, Collina et al.~WINE'20, Kessel et al.~EC'22) studied what the latter referred to as the Stationary Prophet Inequality Problem, due to its similarity to the classic finite-time horizon prophet inequality problem.
    These works all consider unit-demand buyers.
    Mirroring the long line of work on the classic prophet inequality problem subject to \emph{combinatorial constraints},
    we initiate the study of the stationary prophet inequality problem subject to \emph{combinatorially-constrained buyers}.

    Our results can be summarized succinctly as unearthing an algorithmic connection between contention resolution schemes (CRS) and stationary prophet inequalities.
    While the classic prophet inequality problem has a tight connection to online CRS (Feldman et al.~SODA'16, Lee and Singla ESA'18), we show that for the stationary prophet inequality problem, \emph{offline} CRS play a similarly central role.
    We show that, up to small constant factors, the best (ex-ante) competitive ratio achievable for the combinatorial prophet inequality equals the best possible balancedness achievable by offline CRS for the same combinatorial constraints.
\end{abstract}



\pagenumbering{gobble}

\section{Introduction}

A core challenge in economics is tackling the tension between thickening the market and clearing it. 
More broadly, in decision-making under uncertainty, a recurring theme is deciding between using scarce resources now, or keeping these resources for use for unknown later opportunities.

One problem capturing such tension is the (classic) prophet inequality problem, introduced in the 70s by Krengel and Sucheston \cite{krengel1977semiamarts,krengel1978semiamarts}.
Here, a seller has a single item, and buyers arrive sequentially, and then announce their bid for the item, drawn from a priori known distributions. 
When a buyer arrives, the seller must immediately either reject their bid or sell them the item at the bid price, thus forgoing all future buyers' bids.
The objective is to maximize the \emph{competitive ratio}, i.e., the ratio of the algorithm's expected reward to that of the ``prophetic'' offline algorithm, who can foretell the bid realizations.

The above basic problem has been generalized significantly, and (near-)optimal competitive ratios are known for sellers with varying combinatorial constraints, including multi-unit \cite{hajiaghayi2007automated,alaei2012online,arnosti2022tight}, matroid \cite{kleinberg2019matroid}, polymatroid \cite{dutting2015polymatroid}, matching \cite{ezra2020online,gravin2019prophet}, knapsack \cite{feldman2016online}, and arbitrary downward-closed sets \cite{rubinstein2016beyond}.\footnote{Recall that a set family $\calF\subseteq 2^\calG$ is downward-closed if whenever $A\in \calF$ and $B\subseteq A$, then $B\in\calF$.} Similarly, there has been growing interest in generalizing from linear rewards (i.e., sum of bids) to more involved set functions, such as subadditive, submodular and XOS valuations \cite{dutting2020log,rubinstein2017combinatorial,zhang2022improved,dutting2020prophet}.
Finally, some recent work has considered this family of problems from the perspective of poly-time approximation
of the (computationally unbounded) optimal online algorithm \cite{niazadeh2018prophet,anari2019nearly,papadimitriou2021online,braverman2022max,naor2023online,dutting2023prophet}.

Despite the breadth of work on prophet inequality problems, one key aspect of markets was overlooked: over sufficiently long time horizons, the seller may obtain additional goods to sell.

This gap in the literature was addressed by a number of recent works, the most relevant of which to our work are
\cite{aouad2020dynamic,collina2020dynamic,kessel2022stationary}. 
In particular, \cite{kessel2022stationary} introduced the \emph{stationary prophet inequality}, where items of different goods are produced and perish according to a Poisson time process, and buyers of different preference types similarly arrive according to a Poisson point process, and the algorithm must immediately decide what good to sell an item of to each buyer on arrival. 
The goal is to maximize the algorithm's infinite-time-horizon average reward. (See \Cref{sec:prelims}.)
Kessel et al.~\cite{kessel2022stationary} provide optimal $\frac{1}{2}$-competitive algorithms for the single-good problem, and algorithms improving on the multi-good algorithm by \cite{collina2020dynamic}, and the approximation of the optimal online algorithm by \cite{aouad2020dynamic} for the single-good problem.\footnote{It should be noted that \cite{aouad2020dynamic,collina2020dynamic} study more broadly a dynamic stochastic matching problem, which allows buyers to depart stochastically but not immediately, and to be sold items at any point until their departure.}

While the preceding papers provide a good first treatment of infinite-time-horizon two-sided markets, these papers are all restricted to unit-demand buyers. 
In this work, we extend this study to \emph{combinatorial-demand} buyers, who are interested in subsets of items belonging to a downward-closed 
set family, and may stipulate non-linear valuations for such sets. (See \Cref{sed:def}.) 
Our goal is to similarly develop a rich theory of achievable competitive and approximation ratios for infinite time horizons, mirroring the rich literature on the (finite time horizon) classic prophet inequality.

\vspace{-0.2cm}
\subsection{Results}

Our main result ties the combinatorial stationary prophet inequality problem to a central tool in the study of this problem's classic counterpart: contention resolution schemes (CRSes).
\begin{Def}[\cite{vondrak2011submodular}] \label{def:CRS}
Let $(\calG,\calF\subseteq 2^{\calG})$
be a set system and let polytope $\calP_\calF$ be a relaxation of $\operatorname{convexhull}\{\mathbbm{1}_F \mid F\in \calF\}$.\footnote{Polytope $P$ is a relaxation of another polytope $Q$ if it contains the same $\{0,1\}$-points, i.e. $P \cap \{0,1\}^n \supseteq Q\cap \{0,1\}^n$.
Throughout, all polytopes we work  with are \emph{solvable}, i.e., there exist efficient separation oracles for these polytopes.
Moreover, to avoid clutter, we often implicitly assume a fixed such polytope $\calP_\calF$ when discussing a constraint $\calF$.}
A \emph{Contention Resolution Scheme (CRS)} $\pi$ for $(\calG,\calF)$
and $\calP_\calF$ is an algorithm that takes as input a point $\vec x \in \calP_\calF$ and a set of \emph{active} elements 
$R(\vec x) \subseteq \calG$, including each element $i \in \calG$ independently with probability $x_i$, and outputs a feasible subset $\pi_{\vec{x}}(R(\vec x))\in \calF\cap 2^{R(\vec x)}$. 
The CRS $\pi$ is $c$-balanced (a.k.a, $\pi$ has \emph{balance ratio} $c$) if for every $\vec{x}\in \calP_\calF$,
$$\Pr [ i \in \pi_{\vec{x}}(R(\vec x)) \mid i \in R(\vec x) ] \geq c.$$
\end{Def}

Contention resolution schemes, originating in the study of submodular maximization \cite{calinescu2011maximizing}, have found wide-ranging applications including mathematical programming \cite{vondrak2011submodular}, combinatorial sparsification \cite{dughmi2022sparsification} and online rounding schemes \cite{naor2023online}, among others. 
\emph{Online} CRSes (OCRSes), that observe online for each element $i$ whether $i\in R(\vec{x})$, and must decide before observing the next element whether $i\in \pi_{\vec x}(R(\vec{x}))$, likewise have far-reaching applications to online optimization and beyond, including 
prophet inequalities \cite{feldman2016online,gravin2019prophet}, stochastic probing \cite{feldman2016online,adamczyk2015improved,baveja2018improved}, oblivious posted pricing mechanisms \cite{pollner2022improved,feldman2016online,ezra2020online} and algorithmic delegation \cite{bechtel2022delegated}. Most relevant to our work is OCRSes' tight connection to the classic prophet inequality problem:
OCRSes yield prophet inequalities with the same balance ratio \cite{feldman2016online}, while \emph{ex-ante} prophet inequalities (yielding an approximation of the ex-ante relaxation) imply OCRSes with the same value, as proven by Lee and Singla \cite{lee2018optimal}.

In contrast to the above, our main result is a proof that \emph{offline} contention resolution plays a similar central role for the (online) stationary prophet inequality problem.

\begin{restatable}{thm}{mainthm}\label{thm:main_thm}
Let $\calF$ be a downward-closed set family for which the maximum achievable CRS balance ratio is $c$.
Then, stationary prophet inequality for $\calF$-constrained buyers admits a (ex-ante) $(c/2)$-competitive algorithm, while no ex-ante $(c+\eps)$-competitive algorithm is possible, for~any~$\eps>0$.
\end{restatable}

At the cost of decreasing the competitive ratio by a factor of $(1-1/e-\eps)^2$, or only $(1-1/e-\eps)$ if we ignore computational aspects (but not uncertainty about the future), our results also extend to monotone \emph{submodular} objectives (and worse constants for arbitrary non-negative submodular functions), provided the CRS used is \emph{monotone}.

The algorithmic part of Theorem~\ref{thm:main_thm} yields a plethora of results for rich families of constraints on buyers' demands (and also extend to settings where different buyers have different constraints). One prominent example is the multi-good setting, where buyers are again unit-demand. 
For this problem, the best known competitive ratio is $0.267$ \cite{kessel2022stationary}, improving on a previous $\frac{1}{8}=0.125$ of \cite{collina2020dynamic}.
Our general result beats these bounds, giving a $\frac{1}{2}(1-1/e)\approx 0.316$-competitive algorithm for this special case.
Our next result further improves on the above bound for this basic problem.

\begin{thm}\label{thm:multi-good-intro}
The multi-good (unit-demand) stationary prophet inequality problem admits an algorithm wich competitive ratio $(1-1/\sqrt{e})\approx 0.393$.
\end{thm}

Finally, we focus our attention on \emph{matroid-constrained} stationary prophet inequalities. For this problem, Theorem~\ref{thm:main_thm} gives a $\frac{1}{2}(1-1/e)$-competitive algorithm. Echoing a growing interest in the question of efficiently approximating the optimum \emph{online} algorithm, we show that (slightly) better approximation is achievable by a polytime algorithm by extending our approach.

\subsection{Technical Overview}

Our starting point is the observation that some prior algorithms for single/multi-good SPI  \cite{collina2020dynamic,kessel2022stationary} can be interpreted as repeatedly applying offline contention resolution schemes when buyers arrive.
In hindsight, this is natural: by the central PASTA property from queueing theory (Lemma~\ref{pasta}), an arriving buyer observes a set of available items drawn from the stationary distribution. 
So, intuitively, ``proposing'' each available good to the buyer independently with probability derived from some LP relaxation, and then applying a CRS to choose proposed good(s) to sell should yield high average reward.

While the above approach sounds simple enough, it does not quite work as stated, since goods' availabilities are not independent, due to contention of different goods for previously-arrived buyers. 
Nonetheless, via stochastic coupling arguments (see Lemma~\ref{lem:stochastic-dominance}), for unit-demand buyers in both single-good and multi-good problems, one can couple the goods' availabilities in this process and those in a process in which goods face no contention, resulting in independent availability. This allowed \cite{collina2020dynamic,kessel2022stationary} to effectively appeal to contention resolution schemes, though sub-optimal ones, and not used in a black-box manner, as the proposals' distribution only \emph{dominates} an independent distribution.
In contrast, our improved multi-good algorithm of Theorem~\ref{thm:multi-good-intro} does rely on contention resolution schemes in a black-box manner, but explicitly relies on optimal CRSes for single items subject to \emph{dependent distributions} (Theorem~\ref{thm:single_choice_CRS}).

Our main divergence from prior work concerns combinatorially-constrained buyers. 
Here, each buyer may be sold items of \emph{multiple} different goods, which can create positive correlations between different items' availabilities, resulting in poor balance ratio of any CRS for the obtained distribution \cite{dughmi2019outer}.
To overcome this challenge, rather having buyers get proposals from available items, we let them get proposals from \emph{present} items -- i.e., items that have not perished, but may have been sold. 
 This then gives us independent distributions (as presence is independent across goods), and allows us to use an optimal CRS for the constraint family in a black-box manner.
Following this step, we then sell the buyer the subset of items selected by the CRS that are also available.
Our lower bounds on our algorithm's average reward (for linear objectives) then follow from linearity of expectation, once we lower bound the probability of a selected (present) item to also be available, which is the crux of the analysis.
For this last step, we essentially reduce to the analysis of the single-item case in \cite{kessel2022stationary}.
Thus, we obtain our main result of Theorem~\ref{thm:main_thm} from a ``two-pronged'' reduction: reducing to black-box use of optimal CRS on the one hand and to the single-good problem on the other.

Our approximation of the optimum online algorithm for matroid-constrained buyers follows the above approach, with the following twist: there, we may afford to propose goods more frequently (by virtue of additional LP constraints), and so our algorithm
requires an extension of CRS for points slightly outside the matroid polytope  (Lemma~\ref{lem:extended-crs}).

Our analysis for submodular objectives is more intricate, as linearity of expectation is no longer applicable. Instead, we lower bound the expected marginal contribution of each element to the set selected (and then to the subset sold) of the CRS, relating it to the multilinear extension of the selling rates $\vec{x}$ that approximately maximize the multlinear extension (using known constant-approximation for such objectives subject to solvable polytopes). 
As the multilienar extension provides a constant approximation for the concave extension, or the highest-valued distribution with given marginals, this yields a constant-approximate SPI, with the constant again depending linearly on the best balance ratio of a CRS for the underlying constraint.

Finally, our asymptotically tight lower bound of Theorem~\ref{thm:main_thm} similarly builds on focusing on allocation to present items, and then obtains its bound by known connections between the highest balance ratio of a CRS and the correlation gap.

\section{Preliminaries}
\label{sec:prelims}

\subsection{Problem statement}\label{sed:def}
In the stationary prophet inequality problem (SPI), a seller wishes to sell items of $n$ types of goods $\calG$ to $m$ types of buyers $\calB$, while (approximately) maximizing the seller's average gain over an infinite time horizon. 
Buyers and items arrive and depart according to independent Poisson point processes, with arrival and departure times only observed by the seller as they occur, as follows.

\underline{\textbf{Items}} of good $i \in \calG$ are homogeneous, supplied according to a Poisson process with rate $\lambda_i>0$, and perish at an exponential rate $\mu_i > 0$.
Thus, in any interval $[t, t + \Delta]$, in expectation $\lambda_i \cdot \Delta$ items of good $i$ are supplied, and an item that is supplied at time $t'$ perishes at a time $t''\sim t'+\Exp(\mu_i)$. 
An item is \emph{present} if it has been supplied but has not yet perished, and a present item is also \emph{available} if it has not been sold to a buyer. 
Similarly, a \emph{good} $i\in \calG$ is present or available if an item of said good is present or available, respectively.
We denote the sets of present and of available goods at time $t>0$ by $P^t\subseteq \calG$ and $A^t\subseteq \calG$, respectively.
We similarly denote by $P=P^\infty$ and $A=A^{\infty}$ the ``stationary'' set of present and available goods.

\underline{\textbf{Buyers}} of type $j\in \calB$ arrive according to a Poisson process with rate $\gamma_j > 0$. 
Each buyer type $j \in \calB$ offers a value for sets of items,  given by valuation function $f_j:\calG\to \mathbb{R}_{\geq 0}$. In the basic case, $f_j$ is linear, $f_j(S) = \sum_{i\in S}v_{ij}$.
Upon arrival of a buyer of type $j\in \calB$, the seller must irrevocably decide which available set of goods $S\subseteq \calG$ to sell an item of to the buyer,
where $S$ must belong to some known downward-closed family $\calF_j\subseteq 2^{\calG}$ (e.g., independent sets in matroids, matchings in graphs, etc). The seller then gains $f_j(S)$ value, and the buyer immediately departs.

\underline{\textbf{The seller}} has as objective the maximization of the average infinite-time reward on the input SPI instance $\calI$. 
We measure the algorithm's reward in terms of (i) its competitive ratio, i.e., its approximation of the average reward $\optoff(\calI)$ of the optimal offline algorithm (the ``prophet''), which knows the arrival and departure times up front or (ii) its approximation of the average reward $\opton(\calI)$  of the optimal online algorithm, which does not know future arrival and departure times, but is computationally unbounded.  

We refer to an SPI instance where buyers are constrained by $\calF$ as $\calF$-SPI for short.

\paragraph{Notation.} Throughout, for a set $\calU$ and vector $\vec x\in [0,1]^{\calU}$, we denote by $\calU(\vec x)$ a random set containing each element $i\in \calU$ independently with probability $x_i$. 

\paragraph{Queuing theory background.}
We recall the following fundamental PASTA property (``Poisson Arrivals See Time Averages''). 
\begin{lemma}[PASTA \cite{wolff1982poisson}]\label{pasta}
The fraction of Poisson arrivals who observe a stochastic process in a state is equal to the fraction of time the stochastic process is in this state, provided that the Poisson arrivals and the history of the stochastic process are independent.
\end{lemma}

The above implies, 
for example, that the fraction of arriving buyers observing that $i$ is present is precisely $\Pr[i\in P]$. The following lemma asserts that this latter probability, $\Pr[i\in P]$, upper bounds the probability of $i$ being present at any given time $t\geq 0$ during the process. For a proof of this standard fact about birth-death processes, see e.g., 
\cite[Proposition 9.2.4]{ross1995stochastic}.

\begin{prop}\label{monotonicity of BD Processes}
$\Pr[i\in P^t]$ is monotonically increasing in $t\geq 0$, and in particular for any $t\geq 0$, $$\Pr[i\in P^t]\leq \Pr[i\in P^\infty]=\Pr[i\in P].$$
\end{prop}

\subsection{Further background on CRS}
A key tool we use in our algorithms are contention resolution schemes of high balance ratio. 
The highest possible such balance ratio of any CRS for a constraint polytope $\calP_\calF$ is known to be tightly related to the \emph{correlation gap} of this polytope, which is the worst-case ratio between the value of the best correlated and independent distributions satisfying a given set of marginals in said polytope. 
\begin{prop}[\cite{chekuri2011multi,agrawal2012price}] \label{thm:CRS-duality}
If the highest balance ratio of a CRS for family of constraints $\mathcal C$ over the set of elements $\mathcal G$ is $\alpha$, then for any $\epsilon >0$ there exists a constraint $\mathcal F \in \mathcal C$, vector $\vec x \in P_\calF\cap [0,\eps]^{|\mathcal G|}$ and weight vector $\vec y$~s.t.
$$ \frac{\mathbb E \left[\max_{S \subseteq R(\vec x), S\in \calF}  \sum_{i\in S} y_i \right]}{\sum_{i\in \calG} y_i\cdot x_i} \leq \alpha .$$
\end{prop}

Finally, we require a CRS for rank-one uniform matroids with \emph{correlated distributions}, which in particular does not require independence across the events $\{i\in R\}_{i\in\calG}$.
\begin{thm}[\cite{feige2006maximizing}]\label{thm:single_choice_CRS}
    Let $\calD$ be a distribution over $2^\calG$ such that for all $S\subseteq \calG$,
    \begin{equation*}
        \Pr_{R\sim \calD}[S\cap R \neq \emptyset ] \geq \sum_{i\in S} \beta_i \cdot \Pr_{R\sim \calD}[i\in R].
    \end{equation*}
    Then there exists a poly-time algorithm that selects a subset $\pi(R)$ of $R\sim \calD$  of size at most one, satisfying 
    $$\Pr[i\in \pi(R)]\geq \beta_i \cdot \Pr[i\in R].$$
\end{thm}

Finally, we also need the notion of \emph{monotone} CRS for our results for SPI with submodular objectives (see \Cref{sec:submodular} for background on submodularity).
\begin{Def}
    A CRS $\pi$ is \emph{monotone} if for every $S\subseteq T \subseteq \calG$,
$$\Pr[i\in \pi(S) \mid i\in S] \geq \Pr[i\in \pi(T) \mid i\in T].$$
\end{Def}

\subsection{The (ex-ante) relaxations}
\label{sec:prior-LP-benchmarks}

So far, previous works \cite{collina2020dynamic,aouad2020dynamic,kessel2022stationary} proposed LP benchmarks for the multi-good (unit-demand buyers) stationary prophet inequality problem.
We generalize these LPs to buyers with combinatorial demands in this section (and to buyers with submodular valuations in \Cref{sec:submodular}). 

\begin{lemma}
\label{lem:basic-constraints}
Let $x_{ij}$ be the rate at which some algorithm $\calA$ sells items of good $i \in \calG$ to buyers of type $j \in \calB$. Then $\mathbf{x} := (x_{ij})_{i \in \calG, j \in \calB}$ with $\vec x_j := (x_{ij})_{i\in \calG}$ satisfies the following constraints: 
\begin{align}
    \sum_{j \in \calB} x_{ij} & \leq \lambda_i && \forall i \in \calG \label{eqn:flow-constraint-seller} \\
    x_{ij} & \geq 0 && \forall i\in \calG, j\in \calB \label{eqn:positivity} \\
    x_{ij} & \leq \gamma_j \cdot \left(1-\exp\left(-\frac{\lambda_i}{\mu_i}\right)\right) && \forall i \in \calG, j \in \calB \label{eqn:ec-22-constraint} \\
   \frac {\mathbf{x}_{j}}{\gamma_j} & \in \calP_\calF. && \forall j \in \calB \label{eqn:flow-constraint-buyer-comb}
\intertext{If $\calA$ is an \emph{online} algorithm, then $\vec{x}$ also satisfies the following constraint:}
    x_{ij} & \leq \gamma_j \cdot \left( \frac{\lambda_i - \sum_{\ell \in \calB} x_{i\ell}}{\mu_i} \right). && \forall i\in \calG,j\in \calB \label{ec20-constraint}
\end{align}
\end{lemma}
\begin{proof}
Constraints \eqref{eqn:flow-constraint-seller} and \eqref{eqn:positivity} are immediate, since Algorithm $\calA$ cannot sell a good $i\in \calG$ at a rate higher than $i$'s arrival rate, or at a negative rate. 
Next, as used by \cite{kessel2022stationary}, since a buyer of type $j$ can only be sold an item of good $i$ if $i$ is available, and hence present, Constraint \eqref{eqn:ec-22-constraint} follows from the PASTA property (Proposition~\ref{pasta}) and $\Pr[i\in P]=1-\exp\left(-\lambda_i/\mu_i\right)$, (see \Cref{lem:Pr-available-exact}). 
The (new) Constraint \eqref{eqn:flow-constraint-buyer-comb} follows from each arrival of buyer $j$ being sol a set in $\calF$, and so $\frac{\bf{x}_j}{\gamma_j}$ is a distribution over indicator vectors of sets in $\calF$, and thus belongs to the relaxed polytope $\calP_\calF$.
Finally, Constraint \eqref{ec20-constraint} holds for any online algorithm for unit-demand single-good SPI, as proven by \cite{aouad2020dynamic}, and so if we restrict our attention to rates $(x_{ij})_{j\in \calB}$, we find that this constraint holds for combinatorial-demand SPI as well.
\end{proof}

Thus, for any combinatorial stationary prophet inequality instance $\calI$, the above constraints define two relaxations, that are solvable by the ellipsoid method if $\calP_\calF$ is solvable (as assumed).
\begin{align*}
 \calP_\mathrm{off}(\calI) & := \left\{\mathbf{x}\in \mathbb{R}^{\calB\times \calG} \;\middle\vert\; \mathbf{x} \textrm{ satisfies } \eqref{eqn:flow-constraint-seller}-\eqref{eqn:flow-constraint-buyer-comb}\right\}, \\
 \calP_\mathrm{on}(\calI) & := \left\{\mathbf{x}\in \mathbb{R}^{\calB\times \calG} \;\middle\vert\; \mathbf{x} \textrm{ satisfies } \eqref{eqn:flow-constraint-seller}-\eqref{ec20-constraint}\right\}.
\end{align*}

\begin{cor}
\label{cor:RB-geq-OPT}
For all instances $\calI$ of the combinatorial stationary prophet inequality problem,
\begin{enumerate}
    \item $\max\left\{\sum_{i, j} v_{ij} \cdot x_{ij} \;\middle\vert\; \mathbf{x} \in \calP_{\mathrm{off}}(\calI)\right\} \geq \optoff(\calI)$, and
    \item $\max\left\{\sum_{i, j} v_{ij} \cdot x_{ij} \;\middle\vert\; \mathbf{x} \in \calP_{\mathrm{on}}(\calI)\right\}\geq \opton(\calI)$.
    \end{enumerate}
    Both programs on the LHS of the above are solvable in polynomial time if $\calP_\calF$ is a solvable polytope.
\end{cor}
\begin{proof}
The inequalities follow directly from \Cref{lem:basic-constraints}. The polytime solvability follows from the fact that $\calP_\mathrm{off}$ and $\calP_\mathrm{on}$ consists of intersections of polynomial many linear constraints and constraint sets of copies of $\calP_{\calF}$. 
Hence, if $\calP_\calF$ is solvable then $\calP_\mathrm{off}$ and $\calP_\mathrm{on}$ are also solvable.
\end{proof}

\section{Additive Stationary Prophet Inequalities}\label{sec:combinatorial}


In this section, we introduce and analyze our main combinatorial-demand SPI algorithm under linear valuation functions. For notational simplicity, we assume that all buyers have the same combinatorial constraint over goods, $\calF\subseteq 2^{\calG}$, though our analysis extends seamlessly to the case of different constraints for different buyer types, with the competitiveness guarantees modified appropriately.
In what follows we let $\calP_\calF$ be some fixed relaxation of $\text{convexhull}\{\mathds{1}_F \mid F\in \calF\}$.

\subsection{The main algorithm}\label{sec:analysis OPToff Combinatorial}
Our combinatorial stationary prophet inequality algorithm, \Cref{alg:combinatorial}, works as follows.
It takes as input a vector $\mathbf{x}\in \mathbb{R}^{|\calG|\times |\calB|}$,\footnote{For concreteness, this vector can be $\arg\max\{\mathbf{v}\cdot \mathbf{x} \mid \mathbf{x}\in \calP_\mathrm{off}(\calI)\}$.} 
and some per-good scaling vector $\mathbf{w}\in \mathbb{R}^{\calG}$. It further uses a $c$-balanced CRS, $\pi$.
The algorithm then attempts to (approximately) follow the sell rate prescribed by $\mathbf x$ as follows. 
When a buyer of type $j\in \cal B$ arrives at time $t$, 
we have each present good \emph{propose} to buyer $j$ with probability $q_{ij}:= \frac{x_{ij}}{\gamma_j\cdot w_i}$ independently. Denote the set of proposers to buyer $j\in \calB$ by $R_j^t$. 
Then, to allocate \emph{feasible} subsets, we use a CRS $\pi$ for constraint $\calF$ and polytope $\calP_\calF$, and allocate an item of each available good in the selected subset $\pi(R_j^t)\subseteq R_j^t$ to the buyer.
\begin{algorithm}
	\begin{algorithmic}[1]
		\Require $\mathbf{x}\in \mathbb{R}^{|\calG|\times |\calB|}$, $\mathbf{w} \in \mathbb{R}^{|\calG|}$, a CRS $\pi$ for $\calP_\calF$
		\ForAll{arrivals of buyer of type $j \in \calB$ at time $t$}

		\State Let $R^t_j\subseteq \calG$ contain each present good $i \in \calG$ independently with probability $q_{ij} :=  \frac{x_{ij}}{\gamma_j \cdot w_i}$.\label{line:i-proposes-j-comb}
		
		\State Allocate an item of all available goods $i\in\pi(R^t_j)$ to the current buyer. \label{line:Allocation}
		\EndFor
	\end{algorithmic}
	\caption{The combinatorial stationary prophet inequality algorithm}
	\label{alg:combinatorial}
\end{algorithm}

As we will show, invocations of \Cref{alg:combinatorial} with appropriately-chosen $\mathbf{x}$ and $\mathbf{w}$ obtain constant competitive ratio and approximation of the optimum online algorithm 
for a wide range of ``nicely'' behaved combinatorial constraints---those admitting balanced contention resolution schemes. The following is the main theorem of this section. 
\begin{thm}\label{thm:competitive-paper-body}
    \Cref{alg:combinatorial} run on 
    $\vec x^*\in \arg\max\{\vec v\cdot \vec x \mid \vec x\in \calP_\mathrm{off}(\calI)\}$ 
    and equipped with a $c$-balanced CRS for polytope $\calP_\calF$ is a $(c/2)$-competitive algorithm for stationary prophet inequality with $\calF$-constrained buyers.
\end{thm}

For our analysis, we require the following notation.

\paragraph{Algorithm-specific notation.} Recall that $A^t
\subseteq \calG$ and $P^t\subseteq \calG$ denote the sets of available and of present goods at time $t\geq 0$. We let $A$ and $P$ denote the sets of available and present goods at the stationary distribution of the stochastic process induced by \Cref{alg:combinatorial}. 
Similarly, we denote by $R_j$ the set of proposals to a buyer of type $j$ if the buyer arrives at time $t\to \infty$.
\color{black}

Before laying the groundwork for analyzing the competitive/approximation ratio of \Cref{alg:combinatorial}, we note that it is well-defined for the choices of $\bf{x}$ and $\bf{w}$ we will use.
\begin{lem}\label{lem:alg-correct}
\Cref{alg:combinatorial} is well-defined and outputs a feasible allocation if run on $\bf{x}\in \calP_\mathrm{off}(\calI)$ and $\bf{w}$ with $w_i = 1-\exp(-\lambda_i/\mu_i)$ for all $i\in \calG$.
\end{lem}
\begin{proof}
\Cref{line:i-proposes-j-comb} is well-defined, as $q_{ij}\in [0,1]$, by constraints \eqref{eqn:ec-22-constraint}. Next, we note that presence of different goods is independent, as are the $\Ber(q_{ij})$ coins in \Cref{line:i-proposes-j-comb}, and so $R^t_j$ is drawn from a product distribution over $\calG$ for all buyer $j\in \calB$.
Moreover, upon buyer $j$'s arrival, each present good $i\in \calG$ is added to the set of goods $R_j^t$ with probability $q_{ij}$.
\begin{align*}
         q_{ij} = \frac{x_{ij}}{\gamma_j\cdot w_i}\cdot \Pr[i\in P^t] \leq  \frac{x_{ij}}{\gamma_j\cdot w_i}\cdot \Pr[i \in P] =  \frac{x_{ij}}{\gamma_j},
\end{align*}
where the inequality holds by definition of $q_{ij}$ and Proposition\ref{monotonicity of BD Processes}, while the final equality follows by $\Pr[i \in P] = w_i$ due to Proposition \ref{lem:Pr-available-exact}. Furthermore,  \Cref{eqn:flow-constraint-buyer-comb}
ensures that $\mathbf{x}_j \in \gamma_j\cdot \calP_\calF$, which implies that $\mathbf{z}_j^t\in \calP_\calF$, where $\mathbf{z}_{ij}^t = \Pr[i\in R^t]$.
Consequently, by the properties of a CRS (Definition~\ref{def:CRS}), the CRS $\pi$ outputs a feasible set $\pi(R_j^t)\in \calF$ in \Cref{line:Allocation}.
Therefore, by downward closedness of $\calF$, \Cref{line:Allocation} allocates a feasible subset to each buyer $j\in \calB$ upon arrival, $\pi(R_j^t)\subseteq R$, $\pi(R_j^t) \in \calF$.
\end{proof}

Our competitive and approximation ratios will follow from linearity together with a lower bound on the sell rate of good $i\in \calG$ to buyer $j\in \calB$, which are easily expressed using the PASTA property (Lemma~\ref{pasta}):
\begin{lem}\label{lem:lowerbound_sellingrate}
    Let $s_{ij}$ be the selling rate of good $i\in \calG$ to buyer type $j\in \calB$ by \Cref{alg:combinatorial}. Then,
    $$s_{ij} = \gamma_j \cdot \Pr[i\in A \cap \pi(R_j)].$$
\end{lem}
\begin{proof}
Upon arrival of a buyer of type $j$, \Cref{alg:combinatorial} allocates a set of available goods from $\pi(R_j)\subseteq R_j$. We first observe that the buyer's arrival is independent of the randomness of CRS $\pi$, as well as of the presence and availability of items, and of the $Ber(q_{ij})$ random variables. The PASTA property (Lemma~\ref{pasta}) then ensures that the fraction of time buyer $j\in \calB$ (Poisson arrivals independent of the history due to momemorylessness of Poisson arrivals) observes that good $i\in \calG$ is both available and belongs to set $\pi(R_j)$ is equal to the fraction of time that this latter event occurs. Therefore, the rate $s_{ij}$ at which an item of good $i\in\calG$ is sold to buyer $j\in\calB$ can be expressed as: 
\begin{equation*}
        s_{ij} =\gamma_j \cdot  \Pr[i\in A \cap \pi(R_j)].\qedhere
\end{equation*}
\end{proof}
Motivated by Lemma~\ref{lem:lowerbound_sellingrate}, we turn to lower bounding $\Pr[i\in A\cap  \pi(R_j)]$.

\subsection{Lower bounding $\Pr[i\in A\cap  \pi(R_j)]$}

In this section we prove a lower bound on $\Pr[i\in A\cap  \pi(R_j)]$ for a class of combinatorial constraints with bounded correlation gap and complete the proof of this section's main theorem. Throughout the section, we set $w_i = \Pr[i\in P] = 1-\exp(-\lambda_i/\mu_i)$ (the last equality follows from Proposition \ref{lem:Pr-available-exact}). Naturally, we take $\mathbf{x}$ to be an optimal solution to $\max\left\{\sum_{i,j} v_{ij}\cdot x_{ij} \mid \mathbf{x}\in \calP_\mathrm{off}(\calI)\right\}.$

\paragraph{Challenging correlations.}  We note that availability of good $i$ may be positively correlated with availability (``presence") of other goods, e.g., if $x_{ij}$ and $x_{i'j}$ are low/high for similar $j$.
This may in turn result in positive correlations between availability of good $i$ and presence (and hence proposals) of other goods. 
Therefore, availability of $i$ may be positively correlated with higher contention (larger $R\setminus\{i\}$), and so $i\in A$ may be \emph{negatively} correlated with selection $i\in \pi(R_j)$. See \Cref{appendix:challenging_correlation}.

The above negative correlation poses a challenge when lower bounding $\Pr[i\in A\cap \pi(R_j)]$.
To overcome this, we study a subset $S\subseteq A$ such that $i\in S$ is \emph{independent} of other goods' presence (and hence proposals), from which we obtain that $\Pr[i\in \pi(R_j) \mid i\in S] = \Pr[i\in \pi(R_j) \mid i\in P]$, 
which allows us to then lower bound $\Pr[i\in A\cap \pi(R_j)]\geq \Pr[i\in S\cap \pi(R_j)]$. We turn to defining $S$.

\begin{Def}[Saved goods and items] \label{def:saved} 
We say each item is \emph{saved} (for future buyers) upon its arrival. 
We say good $i\in \calG$ is \emph{saved} if it has a saved item, and we denote the set of saved goods at time $t$ by $S^t\subseteq \calG$.
Next, we associate each proposal of a saved good $i\in \calG$ to buyer $j\in \calB$ at time $t$ (i.e., $i\in S^t\cap R_j^t$) with a saved item of this good, after which this item is \emph{used} (i.e., is no longer saved), and we sell this proposing item of good $i$ at time $t$ if $i\in \pi(R_j^t)$ upon buyer $j\in \calB$'s arrival at time $t$.
\end{Def}

As we shall see later, the definition of saved items allows us to analyze the probability of a saved good $i$ being selected by the CRS by reducing to the single-good problem, analyzed in \cite{kessel2022stationary}. 
The following observation implies that saved goods are also available, and hence this single-good ``reduction'' is useful to lower bound selection of an available good.

\begin{obs}\label{obs:availibility-saved}
$S^t\subseteq A^t$ at every time $t\geq 0$. 
\end{obs}
\begin{proof}
    Since allocated items are used by the time they are allocated, saved items are available. 
\end{proof}

 Similarly to our use of $A,P$ for the stationary counterparts of $A^t,P^t$, we use $S$ to denote the saved set $S^t$ as $t\to \infty$.
The preceding observation motivates lower bounding the probability of a proposing good being saved (thus lower bounding the probability of a proposing good being available, by Observation~\ref{obs:availibility-saved}), 
as in the following.

\begin{restatable}{lem}{savedgivenproposed}\label{lem:A|R}
 $\Pr[i\in S \mid i\in R_j]\geq \frac{1}{2}$ for all good $i\in \calG$.
\end{restatable}
The proof, which follows by a standard exercise in analysis of birth-death processes, is deferred to \Cref{appendix:combinatorial}. We present a brief sketch below.
\begin{proof}
First, since $(i\in R_j) = (i\in P) \land (X_{ij}=1)$, for $X_{ij}\sim\Ber(q_{ij})$ independent of $i\in P$, together with Bayes' Law, we have that
$\Pr[i\in S \mid i\in R_j] = \Pr[i\in S \mid i\in P]$. So, we wish to lower bound $\Pr[i\in S \mid i\in P]$.
Since $S\subseteq P$, this requires bounding $\Pr[i\in S]$ and $\Pr[i\in P]$. Bounding both terms is a simple exercise in the analysis of birth-death processes' stationary distributions, with the ratio of the two lower bounded by $\frac{1}{2}$ in \cite{kessel2022stationary}. 
\end{proof}

Now, we recall that good $i$ being available (or saved) and  proposing to buyer $j$ is not a sufficient condition for $i$ to be allocated. 
Instead, the good must be available, it must propose, and it must be selected by the CRS $\pi$, i.e., we must have $i\in \pi(R_j)$.
Therefore, we need to bound the probability of $i$ being available (again, we will focus on $i$ being saved) and having a particular competing set of goods $R_j\setminus\{i\}$.
The following lemma gives us such a bound, relying on independence of different goods' presence and of the Bernoulli coin tosses in \Cref{line:i-proposes-j}.

\begin{lem}\label{obs:independence across goods}
    For any good $i\in \calG$ and buyer $j\in \calB$,
    $$\Pr[i\in \pi(R_j) \mid i\in S] =  \Pr[i\in \pi(R_j) \mid i\in P].$$
\end{lem}
\begin{proof}
The event that $i$ is saved ($i\in S$) is determined by previous arrivals and departures of good $i$ and buyers, and outcomes of $\Ber(q_{ij'})$ random variables, all of which are independent of arrivals and departures as well as outcomes of the $\Ber(q_{i'j})$ random variables of all goods in $i'\in \calG\setminus\{i\}$. Therefore $i\in S$ is independent of the proposals of any subset of other goods at time $t$. On the other hand, $S$ is independent from the $\Ber(q_{ij})$ that determines if present good $i\in P$ also proposes, i.e., if $i\in R_j$. 
Thus, 
\begin{align*}
\Pr[(R=T) \land (i \in S)] = 
\Pr[R=T]\cdot \Pr[i\in S \mid i\in R_j] = \Pr[R_j=T]\cdot \Pr[i\in S \mid i\in P]. 
\end{align*}
So, by total probability over deterministic CRS $\tilde \pi$ drawn from the randomized CRS $\pi$, we have
\begin{align*}
   \Pr[i \in \pi(R_j) \mid i\in S] &=  \frac{1}{\Pr[i\in S]} \cdot \sum_{\tilde \pi}\Pr[\pi = \tilde \pi]\cdot \sum_{T: \tilde \pi(T)\ni i} \Pr[(R_j=T) \land (i \in S)] \\
   &=  \frac{1}{\Pr[i\in P]} \cdot \sum_{\tilde \pi}\Pr[\pi = \tilde \pi]\cdot \sum_{T: \tilde \pi(T)\ni i} \Pr[R_j=T]  
   \\
   &= \Pr[i\in \pi(R_j) \mid i \in P].\qedhere 
\end{align*}
\end{proof}

We are now ready to lower bound $\Pr[i\in A\cap \pi(R_j)]$.

\begin{lem}\label{lem:available-selected}
For any good $i\in \calG$ and buyer $j\in \calB$,
\begin{equation*}
    \Pr[i\in A\cap \pi(R_j)] \geq \frac{c}{2} \cdot \frac{x_{ij}}{\gamma_j}.
\end{equation*}
\end{lem}
\begin{proof}
Let $\tilde \pi$ be a deterministic CRS drawn from the randomized CRS $\pi$. 
\begin{align*}
   \Pr[i\in A\cap \pi(R_j)] 
   &\geq \Pr[i\in S\cap \pi(R_j)] && \text{(Observation~\ref{obs:availibility-saved})} \\
   &= \Pr[i\in \pi(R_j)\mid i\in S] \cdot \Pr[i\in S \mid i \in P] \cdot \Pr[i\in P] && \text{(Bayes)}\\
   &= \Pr[i\in \pi(R_j)\mid i\in P] \cdot \Pr[i\in S \mid i \in P] \cdot \Pr[i\in P] && (\text{Lemma~\ref{obs:independence across goods}})\\
   & \geq \frac 1 2 \cdot  \Pr[i\in \pi(R_j)\mid i\in P] \cdot \Pr[i \in P] &&(\text{Lemma~\ref{lem:A|R}})\\
   & \geq \frac 1 2 \cdot  \Pr[i\in \pi(R_j)] &&\text{($\pi(R_j)\subseteq P$)}\\
   & \geq  \frac c 2 \cdot \Pr[i\in R_j] && \text{($\pi$ is $c$-balanced)}\\
   &= \frac c 2 \cdot \frac{x_{ij}}{\gamma_j}. &&  \left(\Pr[i\in R_j]=\frac{x_{ij}}{\gamma_j}\right) && \qedhere
\end{align*}
\end{proof}

With the above we are now ready to prove this section's main result: competitive stationary prophet inequalities from balanced CRS for the same constraints, as stated in Theorem~\ref{thm:competitive-paper-body}.

\begin{proof}[Proof of Theorem~\ref{thm:competitive-paper-body}]
    Let $\vec x^* = \{x_{ij}^*\}_{i,j}$ be the optimal solution to the LP $$\max\left\{\sum_{i, j} v_{ij} \cdot x_{ij} \;\middle\vert\; \mathbf{x} \in \calP_{\mathrm{off}}(\calI)\right\}.$$ Corollary~\ref{cor:RB-geq-OPT} implies that $\optoff(\calI) \leq \sum_{i,j} v_{ij} \cdot x_{ij}^* $. Now combining Lemmas~\ref{lem:alg-correct},  \ref{lem:lowerbound_sellingrate} and \ref{lem:available-selected}, we conclude that running \Cref{alg:combinatorial} with inputs $x^*$, $w_i = 1 - \exp(-\lambda_i / \mu_i)$ and $c$-balanced CRS for $\calF$, the algorithm allocates good $i$ to buyer $j$ at rate $s_{ij} \geq \frac{c}{2} \cdot x^*_{ij}$. The theorem then follows by linearity of expectation and \Cref{cor:RB-geq-OPT}. 
\end{proof}
\begin{rem}
    \Cref{alg:combinatorial} runs in polytime as long as the CRS $\pi$ and selling rates $\vec x^*$ are computable in polynomial time, e.g., if $\calP_\calF$ is solvable \footnote{Above, by a \emph{solvable polytope} we mean a polytope which admits an efficient separation oracle, and in particular any linear objective can be optimized over such, using the ellipsoid method. This includes a polytope expressed via polynomially many linear constraints, matching polytopes, matroid polytopes, the intersections of the previous, etc. } and the CRS $\pi$ is polytime, as known for several constraints including Matroids, Matchings, Knapsack and their intersections \cite{calinescu2011maximizing}. 
\end{rem}

\subsection{Tight connection to offline CRS}
In this section, we show that constant competitive ex-post stationary prophet inequality for constraints $\mathcal F$ implies constant balanced contention resolution schemes for the constraint polytope $\mathcal P_\mathcal F$.
We define ex-ante stationary prophet inequality in the following definition. 
\begin{Def}
An \emph{$\alpha$-approximate ex-ante stationary prophet inequality (ex-ante SPI)} for combinatorial constraints $\calF$ is an SPI algorithm that on SPI instance $\calI$ for $\calF$-constrained buyers has expected value at least
\begin{equation*}
    \alpha \cdot  \max\left\{\sum_{i, j} v_{ij} \cdot x_{ij} \;\middle\vert \; \mathbf{x} \in \calP_{off}(\calI)\right\}.
\end{equation*}
\end{Def}

\begin{thm}
    If class of constraints $\mathcal C$ does not admit $c$-balanced CRS, then there does not exist $c$-competitive SPI for the class of constraints $\mathcal C$.  
\end{thm}
\begin{proof}
    First, for any $\epsilon >0$, by \Cref{thm:CRS-duality}, if $\mathcal C$ does not admit $c$-balanced CRS, then there exist constraints $\mathcal F \in \mathcal C$, vector $\vec z \in P_{\mathcal F}$ with $z_i \leq \epsilon$ for all $i$ and positive weights $\{v_i\}_{ i \in \calG}$ such that \begin{equation}\label{eqn:correlation-gap}
    \frac{\mathbb{E}[\max_{ S \subseteq R(z),S \in \mathcal{F}} v(S)]}{\sum_{i=1}^n v_i \cdot z_i} \leq c.
    \end{equation}
    Consider an SPI instance with $n$ goods  and a single buyer constrained by $\mathcal F$. This buyer arrives at rate $\gamma=1$, and has value $v_i$ for each good $i\in \calG$. Let each good $i$ arrive at rate $\lambda_i = z_i$. Now, 
    we consider the following solution to the LP:
    $$x_i = 1- \exp(-\lambda_i) = 1- \exp(-z_i) \geq z_i-z_i^2 \geq (1-\epsilon)\cdot z_i,$$ where the last inequality holds because $z_i\leq \epsilon$. Moreover, $x_i = 1 - \exp(-z_i) \leq z_i$ and $\vec z \in \calP_\mathcal F$, hence $\vec x \in \mathcal \calP_{\mathcal F}$.
    By linearity of expectation, this solution yields the following lower bound on the LP's optimal solution.
    \begin{align}\label{eqn:LP-lower-bound}
    \max\left\{\sum_{i, j} v_{ij} \cdot x_{ij} \;\middle\vert \; \mathbf{x} \in \calP_{off}(\calI)\right\} \geq \left( 1- \eps \right) \cdot \sum_{i\in G} v_i \cdot z_i.
    \end{align}

Now, we upper bound the performance of any algorithm as follows: we first observe that any algorithm allocates a subset of available goods to the arriving buyer. To relax the problem, we allow algorithms to allocate a subset of \emph{present} goods. For this relaxed problem, an optimal algorithm can be easily characterized as follows: upon arrival of the buyer, the algorithm allocates a subset of feasible present goods $S$ that maximizes $y(S)$, since allocating goods does not impact the presence of goods in the future. Now, the probability of each good $i\in \calG$ being present is $\Pr[i\in P] = 1- \exp(-z_i) \leq z_i$, independently of the presence of any other goods. Therefore, by $P\sim \calG(z)$ and equations \ref{eqn:correlation-gap} and \ref{eqn:LP-lower-bound}, we find that indeed
\begin{align*}
    \mathbb{E}\left[\max_{ S \subseteq P,S \in \mathcal{F}} v(S)\right] = \E\left[\max_{S\subseteq \calG(\vec z), S\in \calF} v(S) \right] & \leq c \cdot \sum_{i\in \calG} v_i \cdot z_i \leq \frac{c}{(1-\epsilon)}\cdot  \operatorname{OPT(LP)}. \qedhere
\end{align*}
\end{proof}

\section{Submodular Stationary Prophet Inequalities}\label{sec:submodular}
In this section, we prove that \Cref{alg:combinatorial} leads to constant competitive SPI for a wide class of constraints when buyer's utilities are submodular functions. 

Here, we recall that a set function $f:2^{\calU}\to \mathbb{R}_{\geq 0}$ is submodular if it captures the notion of diminishing returns. More formally, denoting by $f(i\mid A):=f(A\cup \{i\}) - f(A)$ the marginal contribution of $i$ to set $A\subseteq \calU$, we have the following.

\begin{Def}
A set function $f:2^{\calU}\to \mathbb{R}$ is \emph{submodular} if for every $i\in \calU$ and $A\subseteq B\subseteq \calU$,
$$f(i\mid A) \geq f(i\mid B).$$
\end{Def}

To generalize our results to submodular valuations, we first need to generalize our relaxations of the optimal algorithms' average reward. For linear valuations, this is trivial (given Lemma~\ref{lem:basic-constraints}), using linearity. 
For more general valuation functions $f_j:\calG\to \mathbb{R}_{\geq0}$, we need some way to capture the value for fractional rates $x_{ij}/\gamma_j\in [0,1]^{\calG}$, and agreeing with the value if 
$x_{ij}/\gamma_j\in \{0,1\}^{\calG}$. 
That is, we need some \emph{extension} of the set function to real vectors. For this, our starting point will be the maximum concave extension of $f$, also known as its \emph{concave closure}.

\begin{Def}\label{def:concave_closure}
    The \emph{concave closure} $f^+:[0,1]^{\calU}\to \mathbb{R}_{\geq 0}$ of set function $f:2^{\mathcal U} \rightarrow \mathbb R_{\geq 0}$ is given by
    \begin{equation*}
        f^+(\vec x) = \max \left\{ \sum_{A\subseteq \calG }\alpha_A \cdot f(A)\; \;\middle\vert\; \sum_A \alpha_A=1, \alpha_S \geq 0  \text{ and }   \sum_{A:j\in A} \alpha_A \leq x_j\;\;\forall j \right\}. 
    \end{equation*}
\end{Def}
\begin{prop}[\cite{feige2006maximizing}]
    The concave closure is (as its name suggests) concave.
\end{prop}
Unfortunately, it is NP-Hard to optimize $f^+(\vec x)$ over (even uniform) matroid polytopes \cite{vondrak2008optimal}. 
A more easily approximable extension is the \emph{multilinear extension}, which is the value obtained from a random set drawn from a product distribution with marginals $\vec x$.

\begin{Def}
   The \emph{multilinear extension} $F:[0,1]^{\calU}\to \mathbb{R}_{\geq 0}$ of set function $f:2^{\mathcal U} \rightarrow \mathbb R_{\geq 0}$~is~given~by $$F(\vec{x}):=\E[f(\calU(\vec{x})] = \sum_{S\subseteq \calU} f(S)  \prod_{i\in S}  x_i \prod_{i\not\in S}(1-x_i).$$
\end{Def}

Our interest in the multilinear extension is that it is efficiently approximable subject to solvable polytopes \cite{calinescu2011maximizing,vondrak2008optimal,vondrak2013symmetry} and that moreover it approximately upper bounds the concave closure for submodular functions \cite{vondrak2007submodularity,rubinstein2017combinatorial}. (We make these statements precise later, as we use them.) This will allow us to analyze \Cref{alg:combinatorial} for submodular utilities. 

We start by using the concave closure to upper bound the optimal average reward, as follows. We emphasize that our upper bound relies on the concavity of $f^+$.

 \begin{restatable}{lem}{supermoduar}
     Let $\mathcal A$ be a stationary prophet inequality algorithm with selling rates $\vec x:=(x_{ij})_{i\in \calG,j\in \calB}$. Then the expected average gain of $\mathcal A$ is at most 
     \begin{equation*}
         \sum_{j\in \calB} \gamma_j \cdot f^+_j \left(\frac{1}{\gamma_j} \cdot \mathbf x_j \right).
     \end{equation*}
 \end{restatable}
 \begin{proof}
Let $\tau_j(T)$ be the set of arrival points of buyer $j\in \calB$ by time $T$, and let $n_j(T) = |\tau_j(T)|$. For any $t\in \tau_j(T)$, let $T^j_t$ be the set of goods allocated by Algorithm $\cal A$ to the buyer of type $j\in \calB$ arriving at time $t\in \tau_j(T)$.

Fix a buyer $j\in \calB$. We observe that for $T>0$, $n_j(T) \sim \operatorname{Poi}(\gamma_j T)$. By standard concentration bounds for Poisson random variables \cite[Proposition 11.15]{goldreich2017introduction}, $\Pr[|n_j(T) -\gamma_jT | \leq \gamma_j T^{\frac{2}{3}}] \geq 1- O\left(\exp(-\gamma_j\cdot  T^{1/3})\right)$.  We define the event $\mathcal E_j: = \mathds{1}[|n_j(T) -\gamma_j | \leq \gamma_j T^{\frac{2}{3}}].$ 
As $f^* = \max_{S\subseteq \calG} f(S) <\infty$, and $\Pr[\mathcal E_j^c]$ is exponentially small in $T$, taking $T\rightarrow \infty$, we can focus on the event when $\mathcal{E}_j$ occurs, as follows.  
    \begin{align*}
      \liminf_{T\rightarrow \infty}   \E \left[\frac{1}{T}  \cdot\sum_{t \in \tau_j(T)} f_j(T^j_t)  \right]
      &=\liminf_{T\rightarrow \infty} \frac 1 T\E \left[ \sum_{t \in \tau_j(T)} f^+_j(\mathbbm 1_{T^j_t}) \;\middle\vert\; \mathcal E_j \right] && (\liminf_{T \rightarrow \infty } \Pr[\mathcal E_j] = 1.)\\
        & \leq \liminf_{T\rightarrow \infty} \frac 1 T\E \left[ n_j(T) \cdot  f^+_j\left(\frac{1}{n_j(T)}\sum_{t \in \tau_j(T)} \mathbbm 1_{T^j_t} \right) \;\middle\vert\; \mathcal E_j \right] && (f_j^+ \text{ is concave})\\
        &= \liminf_{T\rightarrow \infty}\E \left[ \gamma_j \cdot  f^+_j\left(\frac{\sum_{t \in \tau_j(T)} \mathbbm 1_{T^j_t}}{\gamma_j \cdot T} \right) \;\middle\vert\; \mathcal E_j \right]\\
        & \leq \liminf_{T\rightarrow \infty} \gamma_j \cdot  f^+_j\left(\frac 1 {\gamma_j}\E \left[\frac{ \sum_{t \in \tau_j(T)} \mathbbm 1_{T^j_t}}{T} \;\middle\vert\; \mathcal E_j  \right] \right). && (f_j^+ \text{ is concave})\\
        &= \gamma_j \cdot f_j^+ \left( \frac{1}{\gamma_j} \cdot \vec x \right).
    \end{align*}
    Above, the last equality holds because $\liminf_{T\rightarrow \infty}\E \left[\frac{ \sum_{t \in \tau_j(T)} \mathbbm 1_{T^j_t}}{T} \;\middle\vert\; \mathcal E_j  \right] = \frac {\vec x} {\Pr[\mathcal E_j]} $.
 \end{proof}

Using the above lemma combined with constraints from Lemma~\ref{lem:basic-constraints}, we define natural relaxations of the combinatorial stationary prophet inequalities with submodular valuations. 

\begin{lem}\label{lem:submodular-relaxation}
    For all instances $\calI$ of the combinatorial stationary prophet inequality problem,
    \begin{enumerate}
        \item $\max \left\{ \sum_{j \in \calB} \gamma_j \cdot f_j^+\left(\frac 1 {\gamma_j} \mathbf x_j \right) \;\middle\vert\; \mathbf{x} \in \calP_\mathrm{off}(\calI)\right\} \geq \optoff(\calI)$, and
    \item $\max\left\{\sum_{j \in \calB} \gamma_j \cdot f_j^+\left(\frac 1 {\gamma_j} \mathbf x_j \right) \;\middle\vert\; \mathbf{x} \in \calP_{on}(\calI)\right\}\geq \opton(\calI)$.
    \end{enumerate}
\end{lem}

We now turn to analyzing \Cref{alg:combinatorial} for submodular SPI instances.
We start by relating the gain of an average buyer of type $j\in \calB$ with the buyer would obtain if they could obtain value from all proposing buyers, which will later tie back to the multilinear extension.

Next, we prove the following key lemma which lower bounds the gain of any buyer in terms of the gain the buyer would obtain if they could accept all proposing items. 

\begin{lem}\label{lem:subd lowerbound at stationary dist}
    For any buyer $j\in \calB$ with monotone submodular utility $f_j$,
    \begin{equation*}
        \mathbb E[f_j(A\cap \pi(R_j))] \geq  \frac{c}{2} \cdot \mathbb E[f_j(R_j)]. 
    \end{equation*}
\end{lem}
\begin{proof} 
We will prove the following, which implies the lemma by summation over $i$ and telescoping.
$$\E[f_j(A\cap \pi(R_j) \cap [i]) - f_j(A\cap \pi(R_j) \cap [i-1])] \geq \frac c 2 \cdot \left(\E[f_j(R_j \cap [i]) - f_j(R_j \cap [i-1])]\right).$$
And indeed, we can lower bound the  LHS, which can be restated as follows, by using monotonicity, and hence positivity of marginals of $f_j$.\footnote{As we show below, a standard pruning process allows us to extend the following to non-monotnoe functions.}

\begin{align*}
    & \E[\mathbbm 1_{i\in A \cap \pi(R_j)} \cdot f_j(i\mid A \cap \pi(R_j) \cap [i-1])]\\
    \geq & \E[\mathbbm 1_{i\in S \cap \pi(R_j)} \cdot f_j(i\mid A \cap \pi(R_j) \cap [i-1])] & A\supseteq S \\
    \geq & \E[\mathbbm 1_{i\in S \cap \pi(R_j)} \cdot f_j(i\mid R_j \cap [i-1])] & f \textrm{ submod., } R_j\supseteq A\cap \pi(R_j) \\
    =&\E [\mathbbm 1_{i\in S} \cdot \mathbbm 1_{i\in \pi(R_j)} \cdot f_j(i\mid R_j \cap [i-1]) \mid i\in R_j ] \cdot \Pr[i\in R_j] & \textrm{Bayes} \\
    =& \E[\mathbbm 1_{i\in S} \mid i\in R_j ] \cdot \E [\mathbbm 1_{i\in \pi(R_j)} \cdot f_j(i\mid R_j \cap [i-1]) \mid i\in R_j ] \cdot \Pr[i\in R_j] & (i\in R_j), (i\in S) \;\bot \; (R_j\setminus \{i\})\\
    \geq& \frac{1}{2}  \cdot \E [\mathbbm 1_{i\in \pi(R_j)} \cdot f_j(i\mid R_j \cap [i-1]) \mid i\in R_j ] \cdot \Pr[i\in R_j] & \textrm{Lemma~\ref{lem:A|R}}\\
    \geq& \frac{1}{2}  \cdot \E [\mathbbm 1_{i\in \pi(R_j)} \mid i\in R_j ] \cdot \E[ f_j(i\mid R_j \cap [i-1]) \mid i\in R_j ] \cdot \Pr[i\in R_j] & \pi \textrm{ is monotone}, f \textrm{ is submod.} \\
    \geq & \frac{c}{2}\cdot  \left(\E[f(R_j \cap [i]) - f(R_j \cap [i-1])]\right).  & \pi \textrm{ is $c$-balanced}
\end{align*}
Essentially all steps of the above derivation are proven inline, but we elaborate on two of these:
The second equality follows because the events $\{i\in R_j\}, \{i\in S\}$ and the random set $R_j\cap [i-1]$ are independent. Notice that the positive correlation of $i\in \pi(R_j)$ and $f_j(i\mid R_j\cap [i-1])$ (conditioned on $i$), as these are decreasing events in the (independent) indicators $[i'\in R_j \mid i\in R_j]$, by monotonicity of the CRS $\pi$ and submodularity of $f$. 
Therefore,  the second-to-last inequality follows by Harris' inequality.\footnote{Harris' inequality states that for any two increasing functions $f,g$ on a partially ordered product probability space (i.e., independence of each element) $\Omega$, such as in the space $\Omega=[R_j \mid i\in R_j]$, we have $$\mathbb E[f\cdot g] \geq \E [f]\cdot \E[g].$$} 
\end{proof}

\begin{rem}
    For non-monotone utilities, we can run a pruning process $\eta$ on the allocated set $\pi(R)\cap A$ similar to \cite[Theorem 1.3] {vondrak2011submodular} such that for the ordered set of goods $\calG=[n]$ satisfies $f_j(i\mid \eta(\pi(R)\cap A) \cap [i-1])\geq 0$. 
    We further modify \Cref{alg:combinatorial} for non-monotone submodular utilities by allocating the set $\eta(\pi(R)\cap A)$ to the arrived buyer. Following the arguments from Lemma~\ref{lem:subd lowerbound at stationary dist}, we obtain
    \begin{equation*}
        \E[f_j(\eta(A\cap \pi(R_j)) \cap [i]) - f_j(\eta(A\cap \pi(R_j)) \cap [i-1])] \geq \frac c 2 \cdot \left(\E[f_j(R_j \cap [i]) - f_j(R_j \cap [i-1])]\right).
    \end{equation*}
    This implies that for non-monotone submodular utilities, $\mathbb E[f_j(\eta(A\cap \pi(R)))] \geq \frac c 2 \cdot \mathbb E[f_j(R)].$
\end{rem}

Next, we relating the gain from proposing items (if one could accept them all) to the ex-ante relation, we can use the preceding lemma to obtain the following lower bound on the expected reward gained from different buyer types.
\begin{lem}\label{lem:buyer's submodular utility}
   For any buyer $j\in \calB$ and $T>0$:
   \begin{equation*} 
       \liminf_{T\rightarrow \infty} \frac 1 T \cdot  \mathbb E \left[ \sum_{t \in \tau_j(T)} f_j( A^t \cap R_j^t)\right] \geq \frac{c}{2}\cdot \gamma_j\cdot F_j\left(\frac{1}{\gamma_j}\vec{x}_j\right). 
    \end{equation*}
\end{lem}
\begin{proof}
For any buyer $j\in \calB$,
    \begin{align*}
        \liminf_{T\rightarrow \infty} \frac 1 T \cdot   \mathbb E \left[ \sum_{t\in \tau_j(T)}f_j(A^t \cap \pi(R_j^t)) \right] &= \gamma_j \cdot \E \left[ f_j(A\cap \pi(R_j)) \right] && \text{(PASTA Lemma~\ref{pasta})}\\
        &\geq \frac c 2 \cdot \gamma_j \cdot \mathbb E \left[ f_j(R) \right] && \text{(Lemma~\ref{lem:subd lowerbound at stationary dist})}\\
        & = \frac c 2 \cdot \gamma_j \cdot F_j \left( \frac 1 {\gamma_j} \mathbf x_j\right). && (\textrm{Def. $R_j$ and $F_j(\cdot)$)} \qedhere 
    \end{align*}
\end{proof}    

The preceding lemma, together with solvability of our relaxation and known correlation gap results yield this section's main result, given by the following theorem.
\begin{thm}
    \Cref{alg:combinatorial} run with a $c$-balanced monotone CRS for polytope $\calP_\calF$ yields the following SPI for $\calF$-constrainted buyers:
    \begin{enumerate}
        \item A polytime $\frac c 2 \cdot (1-1/e-\eps)^2$-competitive SPI for monotone submodular utilities.
        \item A polytime $\frac c 2 \cdot \frac{0.3}{200}$-competitive SPI for non-negative submodular utilities.
        \item A $\frac c 2 \cdot (1-1/e-\eps)$-competitive SPIs for monotone submodular utilities. 
        \item A $\frac c 2 \cdot \frac 1 {200}$-competitive SPIs for non-negative submodular utilities. 
    \end{enumerate}
    \end{thm}
\begin{proof}
    By \cite{vondrak2008optimal}, for any $\eps>0$, there exists a polytime algorithm that $\left(1-\frac{1}{e}-\eps\right)$-approximates the multilinear extensions of a monotone non-negative submodular function subjec to any solvable poltyope, e.g., $\calP_\mathrm{off}(\calI)$,  provided that $\calP_\calF$ is solvable.
    On the other hand, by the correlation gap result of \cite{vondrak2007submodularity}, for every $\vec x\in [0,1]^n$, we have $F(\vec x) \geq \left( 1 - \frac 1 e \right) \cdot f^+(\vec x)$. Therefore, by applying the algorithm of \cite{vondrak2008optimal} to approximately maximize $\max \left\{ \sum_{j \in \calB} F_j\left(\frac 1 {\gamma_j} \mathbf x_j \right) \;\middle\vert\; \mathbf{x} \in \calP_\mathrm{off}(\calI)\right\}$, we obtain a polynomial time algorithm that computes $\vec{x}\in \calP_\mathrm{off}(\calI)$ such that
    \begin{align}\label{eq:find_nearoptimal_rates_monotonesubmodular}
    \sum_{j\in \calB} \gamma_j \cdot F_j \left ( \frac 1 {\gamma_j} \mathbf x_j\right) &\geq \left(1-\frac{1}{e}\right)^2 \max \left\{ \sum_{j \in \calB} f^+_j\left(\frac 1 {\gamma_j} \mathbf x_j \right) \;\middle\vert\; \mathbf{x} \in \calP_\mathrm{off}(\calI)\right\} \notag \\
    &\geq \left( 1- \frac 1 e\right) \cdot \left(1 - \frac 1 e -\epsilon \right) \optoff(\calI).  & \text{(Lemma~\ref{lem:submodular-relaxation})}
    \end{align}
    Therefore, by running \Cref{alg:combinatorial} with selling rates obtained in Equation~\ref{eq:find_nearoptimal_rates_monotonesubmodular} and $c$-balanced monotone CRS for $\calF$ and relaxation $\calP_\calF$ we obtain the claimed competitive ratio of the first part of the theorem, as the obtained algorithm has average reward at least 
    \begin{align*}
         \liminf_{T\rightarrow \infty} \frac 1 T \cdot \sum_{j\in \calB} \mathbb E \left[ \sum_{t \in \tau_j(T)} f_j( A^t \cap R_j^t)\right] &\geq \frac{c}{2} \cdot \sum_{j\in \calB} \gamma_j \cdot F_j\left(\frac{1}{\gamma_j} \vec x_j \right) &&\text{(Lemma~\ref{lem:buyer's submodular utility})}\\
         &\geq \frac c 2 \left(1 - \frac 1 {e}\right)^2 \cdot \optoff. &&\text{(Equation~\ref{eq:find_nearoptimal_rates_monotonesubmodular})}
    \end{align*}

    The remaining parts of the theorem follow in much the same way as the first, with minor modifications.
    For example, for non-negative submodular functions, the best known approximation of the multilinear extension subject to solvable polytopes is $0.3$ \cite{vondrak2013symmetry} and using that for any $\vec x\in [0,1]^n$, we have $F(\vec x /2) \geq \frac 1 {200} \cdot f^+(\vec x)$ \cite{rubinstein2017combinatorial}. 
    Similar to part one, applying the algorithm of \cite{vondrak2013symmetry} to approximately maximize $\max \left\{ \sum_{j \in \calB} F_j\left(\frac 1 {\gamma_j} \mathbf x_j \right) \;\middle\vert\; \mathbf{x} \in \calP_\mathrm{off}(\calI)\right\}$ and using $\vec{x}/2$ as an input and a monotone CRS $\pi$ for $\calF$ yields a polynomial time algorithm for non-negative submodular SPIs that is $\frac c 2 \cdot \frac{0.3}{200}$ approximate.
    
    Finally, the computationally unbounded algorithms follow the exact logic of the preceding algorithms, but avoid the use of the polytime approximation algorithms of \cite{vondrak2008optimal} and \cite{vondrak2013symmetry}, and instead use a vector $\vec{x}$ exactly maximizing the multilinear extension. Hence, we can compute selling rates $\vec x$ (with unbounded computations) which approximate $\sum_{j\in \calB} \gamma_j \cdot F_j \left ( \frac 1 {\gamma_j} \mathbf x_j\right)$ within $\left(1 - 1/e -\epsilon \right)$ and $\frac{1}{200}$. This concludes the proof of the last two parts of the theorem.
\end{proof}
\section{Approximating Online OPT for Matroids}\label{sec:Opton}

Relying on optimal CRS for matroids \cite{vondrak2011submodular}, \Cref{alg:combinatorial} yields a $\frac{1}{2}(1-\frac1e)$-competitive matroid-SPI i.e., it yields a $\frac{1}{2}(1-\frac1e)$ approximation of the optimum offline algorithm. 
In this section, we show that \Cref{alg:combinatorial} (alebit with different parameters) yields a better approxiomation (in polytime) of the optimum online algorithm for matroid-SPI. 
\begin{thm}\label{thm:online-OPT-matroid}
    There exists a polynomial-time algorithm for SPI with matroid-constrained buyers that $\left(\frac 1 2 \cdot \left( 1- \frac 1 e\right) + \epsilon_0\right)$-approximates the optimum online algorithm, for some $\epsilon_0 \geq 0.0019$.
\end{thm}

To motivate our algorithm, we first recall that by (the proof of) Lemma~\ref{lem:alg-correct}, \Cref{alg:combinatorial} run with $w_i=1-\exp(-\lambda_i/\mu_i)$ (which is equal to $\Pr[i\in P]$, by Proposition \ref{lem:Pr-available-exact}) results in each good $i$ proposing independently to buyer $j$ with probability 
$\Pr[i\in R_j]\leq x_{ij}/\gamma_j$, and so by Constraint \ref{eqn:flow-constraint-buyer-comb} the set of proposers is chosen from a product distribution with marginals in $\calP_\calF$. This allows us to appeal to CRS for the constraint set $\calF$. 

To approximate the optimum \emph{online} algorithm, we take $\vec x$ to be a solution to $\calP_\mathrm{on}(\calI)$, and note that we can now pick any smaller $w_i\geq \min(1-\exp(-\lambda_i/\mu_i),(\lambda_i-\sum_{\ell\in \calB}x_{i\ell})/\mu_i)$, and let goods propose with probability $x_{ij}/(\gamma_j\cdot w_i)\leq 1$, where the inequality follows by Constraint \eqref{ec20-constraint}.
However, proposing with higher marginal probability than $x_{ij}/\gamma_j$ results in a vector of proposals that can exceed the relaxation polytope $\calP_\calF$.
A key ingredient in the proof of Theorem~\ref{thm:online-OPT-matroid} are  therefore \emph{extended} CRS for matroids, yielding a bound on the probability of selecting an item even if the input vector (slightly) exceeds the relaxed polytope.

\begin{restatable}{lem}{extendedCRS}\label{lem:extended-crs}
Let $\calM=(\calG,\calF)$ be a matroid and let $\vec{x} \in [0,1]^{|\mathcal G|}\cap c \cdot \calP_\calF$ for some $c\geq 1$.
Then, there exists a function $\pi_x: 2^\calG \rightarrow \calF$ that satisfies the following, for $R\sim \calG(\vec{x})$. 
\begin{enumerate}
    \item $\pi_x(R) \subseteq R$,
    \item $\pi_x(R) \in \calF$, and
    \item $\Pr[i \in \pi_x(R) \mid i \in R] \geq \frac{1}{c} \cdot (1-e^{-c})$.
\end{enumerate}    
\end{restatable}
The proof of Lemma~\ref{lem:extended-crs} is similar to the proof of the correlation gap of monotone submodular functions from \cite{calinescu2011maximizing} and is presented in \Cref{appendix:OPTon} for completeness.

Given the above, we want the marginal proposal probabilities to lie within a slightly scaled up version of the polytope $\calP_\calF$, else we would obtain too little from the extended CRS. 
Consequently, we wish to avoid letting goods with presence probability much larger than $\frac{\lambda_i - \sum_{\ell \in \calB} x_{i\ell}}{\mu_i} $ which leads marginal probability of proposals farther away from the polytope $\calP_\calF$.
We will therefore use different forms for $w_i$ for goods $i$, depending on whether their presence probability is higher or lower than $c \cdot \left( \frac{\lambda_i - \sum_{\ell \in \calB} x_{i\ell}}{\mu_i} \right)$, for $c\geq 1$ optimized later.
We therefore define the following partition of the goods.
\begin{enumerate}
    \item $\calG_L:=\left\{i\in \calG \;\middle\vert\; \Pr[i\in P] \leq c \cdot \left( \frac{\lambda_i - \sum_{\ell \in \calB} x_{i\ell}}{\mu_i} \right)\right\}$,
    \item $\calG_H:=\left\{i\in \calG \;\middle\vert\; \Pr[i\in P] > c \cdot \left( \frac{\lambda_i - \sum_{\ell \in \calB} x_{i\ell}}{\mu_i} \right)\right\} = \calG\setminus \calG_L$.
\end{enumerate}

We now define our variant of \Cref{alg:combinatorial} used in this section, and show that it outputs valid allocations.
\begin{lem}\label{lem:alg-correct-opton}
For any matroid SPI instance $\calI$, consider \Cref{alg:combinatorial} run on $\bf{x}\in \calP_{on}(\calI)$ with the extended CRS $\pi$ from Lemma~\ref{lem:extended-crs} and $\bf{w}$ with $$w_i = \begin{cases}\min \left( \Pr[i\in P], \frac{\lambda_i - \sum_{\ell \in \calB} x_{i\ell}}{\mu_i} \right) & i\in \calG_L\\
\Pr[i\in P]/c & i\in \calG_H.
\end{cases}$$
Then, this algorithm is well-defined and outputs a feasible allocation for each arriving buyer.
\end{lem}
\begin{proof}
Fix time $t\geq 0$ and buyer $j$. We observe that $w_i \geq {\Pr[i\in P]/c}$ for all $i\in \calG$. It is clear from the definition that above holds for $i\in \calG_H$. For $i\in \calG_L$, if $w_i = \Pr[i\in P]$ then again $w_i \geq {\Pr[i\in P]/c}$ as $c\geq 1$. Next, if $w_i = \frac{\lambda_i - \sum_{\ell \in \calB} x_{i\ell}}{\mu_i}$ then by the definition of $\calG_L$, we have $w_i \geq {\Pr[i\in P]}/{c}$.

First we need to show that $q_{ij} \in [0,1]$. Clearly $q_{ij}\geq 0$. Now, for $i\in \calG_L$, $q_{ij} = \frac{x_{ij}}{\gamma_j \cdot w_i} \leq 1$, since $\vec x$ satisfies Constraints~\eqref{eqn:ec-22-constraint} and \eqref{ec20-constraint}. For $i\in \calG_H$, we similarly have:
\begin{equation*}
    q_{ij} = \frac{c \cdot x_{ij}}{\gamma_j \cdot \Pr[i\in P]} \leq \frac c {\Pr[i\in P]} \cdot \left( \frac{\lambda_i - \sum_{\ell \in \calB} x_{i\ell}}{\mu_i} \right) < 1.
\end{equation*}
 
Next, we note that presence of different goods is independent, as are the $\Ber(q_{ij})$ coins in \Cref{line:i-proposes-j-comb}, and so $R^t_j$ is drawn from a product distribution over $\calG$.
We argue that the marginalls of this distribution at each time $t\geq 0$, denoted by
$z^t_{ij} := \Pr[i \in R^t_j]$, 
lie within $c \cdot\calP_\calF$.
For any good $i\in \calG$, by Proposition~\ref{monotonicity of BD Processes} and $w_i \geq \Pr[i\in P]/c$,
\begin{align*}
         z^t_{ij} = \frac{x_{ij}}{\gamma_j\cdot w_i} \cdot \Pr[i\in P^t]\leq  \frac{ x_{ij}}{\gamma_j\cdot w_i} \cdot \Pr[i\in P]\leq c \cdot \frac{x_{ij}}{\gamma_j}.
\end{align*}

On the other hand, Constraint~\eqref{eqn:flow-constraint-buyer-comb}
ensures that $\mathbf{x}_j \in \gamma_j\cdot \calP_\calF$, and so $\mathbf{z}_j^t\in c\cdot \calP_\calF$.
Consequently, by properties of the extend CRS (Lemma~\ref{lem:extended-crs}), $\pi$ outputs a feasible set $\pi(R^t_j)\in \calF$ in \Cref{line:Allocation}.
Therefore, by downward closedness of $\calF$, \Cref{line:Allocation} allocates a feasible subset to each buyer, $\pi(R^t_j)\in \calF\cap 2^{R^t_j}$.
\end{proof}

By Lemma~\ref{lem:lowerbound_sellingrate}, the selling rate of items of good $i$ to buyers of type $j$ by \Cref{alg:combinatorial} is $s_{ij} = \Pr[i\in A\cap \pi(R)]$, where $S$ denotes the \emph{saved goods} (see Definition~\ref{def:saved}). Similar to Lemma~\ref{lem:available-selected}, we can lower bound $\Pr[i\in A\cap \pi (R)]$ as follows. 

\begin{claim}\label{claim:allocation probability online OPT}
For any good $i\in \calG$ and buyer $j\in \calB$,
\begin{equation*}
    \Pr[i\in A\cap \pi(R)] \geq \frac 1 c \cdot (1-\exp(-c)) \cdot \frac{\Pr[i\in S ]}{w_i} \cdot \frac{x_{ij}}{\gamma_j}.
\end{equation*}
\end{claim}
\begin{proof}
    Following the proof of Lemma~\ref{lem:available-selected}, we get,
    \begin{align*}
        \Pr[i\in A\cap \pi(R)] &\geq \Pr[i\in S\cap \pi(R_j)] && \text{(Observation~\ref{obs:availibility-saved})} \\
   &= \Pr[i\in \pi(R_j)\mid i\in S] \cdot \Pr[i\in S ] && \text{(Bayes)}\\
   &= \Pr[i\in \pi(R_j)\mid i\in P] \cdot \Pr[i\in S] && (\text{Lemma~\ref{obs:independence across goods}})\\
        &=\Pr[i\in \pi(R) \mid i\in R ] \cdot \frac{x_{ij}}{\gamma_j \cdot w_i}\cdot \Pr[i\in S] && \text{(Bayes)} \\
        &\geq \frac 1 c \cdot (1-\exp(-c)) \cdot \frac{\Pr[i\in S ]}{w_i} \cdot \frac{x_{ij}}{\gamma_j} .
    \end{align*}
    Above, the last inequality holds due to $z_i = \Pr[i\in R_j]\in c \cdot \mathcal P_\mathcal F$ and Lemma~\ref{lem:extended-crs}.
\end{proof}

To complete the proof of Theorem~\ref{thm:online-OPT-matroid}, we need to lower bound the ratio $\frac{\Pr[i\in S]}{w_i}$ for all goods $i\in \calG$. We will analyse the ratio $\frac{\Pr[i\in S]}{w_i}$ for goods in $\calG_L$ and $\calG_H$ separately. 
For goods in $\calG_L$, we leverage  \cite[Claims 4.4 and 4.5]{kessel2022stationary} and obtain the following lemma, whose proof (similar in spirit to the first part of Lemma~\ref{lem:A|R}) is deferred to \Cref{appendix:combinatorial}.

\begin{restatable}{lem}{optonreduction}\label{lem:A|R G_L}
    For $i\in \calG_L$, we have $$\frac{\Pr[i\in S]}{w_i} \geq 0.656.$$
\end{restatable}

On the other hand, the ratio $\frac{\Pr[i\in S]}{w_i}$ for goods in the set $\calG_H$ is slightly more challenging. First, we note that in Lemma~\ref{lem:A|R} we bounded $\frac{\Pr[i\in S]}{\Pr[i\in P]} \geq \frac{1}{2}$ when $q_{ij} = \frac{x_{ij}}{\gamma_j \cdot \Pr[i\in P]}$.  However, for goods in $\calG_H$, $q_{ij} = \frac{c\cdot x_{ij}}{\gamma_j \cdot \Pr[i\in P]}$ for $c\geq 1$ which sells goods more aggressively which creates backward 
pressure on the number of saved items of good $i$. In the next lemma, we precisely capture this trade-off between selling goods aggressively (parametrized by $c$) and their saving for future buyers. More formally, we show that for goods in $\calG_H$, $\frac{\Pr[i\in S]}{\Pr[i\in P]}\geq \frac{1}{1+c}$ which recovers the lower bound obtained in Lemma~\ref{lem:A|R} when $c=1$. We delegate the proof of the lemma to appendix.  

\begin{restatable}{lem}{availablegh}\label{lem:A|R G_H}
    For $i\in \calG_H$, we have $$\frac{\Pr[i\in S]}{w_i} = c\cdot \frac{\Pr[i\in S]}{\Pr[i\in P]} \geq \frac{c}{1+ c}.$$
\end{restatable}

We are now ready to prove Theorem~\ref{thm:online-OPT-matroid}. 
\begin{proof}[Proof of Theorem~\ref{thm:online-OPT-matroid}]
    Let $x^* = \{x_{ij}^*\}_{i,j}$ be the optimal solution to the LP $$\max\left\{\sum_{i, j} v_{ij} \cdot x_{ij} \;\middle\vert\; \mathbf{x} \in \calP_{\mathrm{on}}(\calI)\right\}.$$ Corollary~\ref{cor:RB-geq-OPT} implies that $\opton(\calI) \leq \sum_{i,j} v_{ij} \cdot x_{ij}^* $. We can lower bound the performance of the \Cref{alg:combinatorial} for matroids with parameters $\vec w$ defined in Lemma~\ref{lem:alg-correct-opton} as:
    \begin{align*}
        &\sum_{i\in \calG} \sum_{j\in \calB} \gamma_j v_{ij} \cdot \Pr[i\in \pi(R_j) \cap A]\\
        \geq &\sum_{i\in \calG} \sum_{j\in \calB} \gamma_j v_{ij} \cdot   \frac{\Pr[i\in S ]}{w_i} \cdot \frac{x^*_{ij}}{\gamma_j} \qquad \qquad \qquad \qquad \text{(Claim~\ref{claim:allocation probability online OPT})} \\
        \geq &\frac 1 c \cdot (1-\exp(-c)) \cdot \left( 0.656 \cdot \sum_{i\in \calG_L} \sum_{j\in \calB}  v_{ij} \cdot x^*_{ij}  + \frac{c}{1+c} \cdot  \sum_{i\in \calG_H} \sum_{j\in \calB}  v_{ij} \cdot x^*_{ij} \right),      
    \end{align*}
    where the last inequality follows from Lemma~\ref{lem:A|R G_L} and Lemma~\ref{lem:A|R G_H}. Now setting $c = 1.14$ and \Cref{cor:RB-geq-OPT}, we obtain:
    \begin{equation*}
      \sum_{i\in \calG} \sum_{j\in \calB} \gamma_j v_{ij} \cdot \Pr[i\in \pi(R_j) \cap A] \geq 0.318 \cdot \sum_{i\in \calG} \sum_{j\in \calB}  v_{ij} \cdot x^*_{ij} \geq \left(\frac 1 2 \cdot \left( 1- \frac 1 e\right) + 0.0019\right)\cdot \opton. \qedhere
      \end{equation*}
\end{proof}

\section{Multi-Good Stationary Prophet Inequality}
\label{sec:multi-good}
In this section, we provide an improved algorithm for the multi-good stationary prophet inequality problem, proving the following theorem. 

\begin{thm}\label{thm:multi-good}
There exists a $\left(1 - \frac{1}{\sqrt e}\right)\approx 0.393$-approximate policy for the multi-good unit-demand stationary prophet inequality problem.
\end{thm}

The algorithm for multi-unit SPI is similar to the algorithm for combinatorial SPI defined in \Cref{alg:combinatorial}. The only difference is in the multi-good setting, each ``available" good instead of a ``present" good proposes to the arrived buyer. Then we resolve contention in the proposals using single choice CRS defined in Theorem~\ref{thm:single_choice_CRS} and allocate an item of some good $i\in \calG$ from the set of proposed good.
Its difference compared to previous algoirthms for multi-good SPI is that we (will) use an \emph{optimal} CRS.
The algorithm's pseudocode is given in \Cref{alg:main_new}. 

\begin{algorithm}
	\begin{algorithmic}[1]
		\Require $\mathbf{x}^*\in \mathbb{R}^{|\calG|\times |\calB|}$,  $\mathbf{w} \in \mathbb{R}^{|\calG|}$ and CRS $\pi$.
		\For{arrival of buyer of type $j \in \calB$}
		
		\State  Let $R^t\subseteq \calG$ contain each \textbf{available} good $i \in \calG$ independently with probability $q_{ij} :=  \frac{x_{ij}}{\gamma_j \cdot w_i}$.\label{line:i-proposes-j}
		
		\State Allocate an item of good selected by $\pi(R^t)$ to buyer $j$.\label{line:single-choice-CRS}
		\EndFor
	\end{algorithmic}
	\caption{The multi-good stationary prophet inequality algorithm}
	\label{alg:main_new}
\end{algorithm}

Our competitive ratio for multi-good SPI follow from linearity, together with a lower bound on the sell rate of good $i\in \calG$ to buyer $j\in \calB$, which are easily expressed using the PASTA property (Lemma~\ref{pasta}), as follows.
\begin{lem}\label{lem:lowerbound_sellingrate_singlechoice}
    Let $s_{ij}$ be the selling rate of good $i\in \calG$ to buyer $j\in \calB$ by \Cref{alg:main_new}. Then,
    $$s_{ij} = \gamma_j \cdot \Pr[i\in \pi(R_j)].$$
\end{lem}
\begin{proof}
Upon the arrival of the buyer $j$, \Cref{alg:combinatorial} allocates a random set of available goods $\pi(R_j)\subseteq R_j$. The PASTA property (Lemma~\ref{pasta}) ensures that the fraction of time buyer $j\in \calB$ (Poisson arrivals independent of the history) observes that good $i\in \calG$ is available and belongs to set $\pi(R_j)$ is equal to the fraction of time that this latter event occurs. Therefore, the rate at which an item of good $i\in\calG$ is sold to buyer $j\in\calB$ can be expressed as: 
\begin{align*}
        s_{ij} &=\gamma_j \cdot  \Pr[i\in \pi(R_j)]. 
\end{align*}
\end{proof}

Next, we present a key lemma that shows that the CRS used in \Cref{alg:main_new} (in \Cref{line:single-choice-CRS}) is $\left(1 - \frac{1}{\sqrt e} \right)$-balanced. 
\begin{lem}\label{lem:lowerbound_on_no_proposal}
For any subset of goods $S\subseteq \calG$,
$\Pr[S\cap R_j \neq  \emptyset] \geq \left(1 - \frac 1 {\sqrt e} \right)\cdot \sum_{i\in S}\frac{x_{ij}}{\gamma_j}.$
Moreover, for all $i\in \calG$, $\Pr[i \in \pi(R_j)] \geq \left( 1 -\frac{1}{\sqrt e} \right)\cdot \frac{x_{ij}}{\gamma_j}$. 
\end{lem}

We note that lower bounding the probability $\Pr[S\cap R \neq  \emptyset]$ directly is challenging, as the availability of good $i\in \mathcal G$ may depend on the availability of other goods. In fact, for any goods $i,i'$, events $i\in A$ and $i'\in A$ can be negatively correlated. 
Earlier works \cite{collina2020dynamic,aouad2020dynamic,kessel2022stationary} resolved unfavorable correlation issues by considering several independent Poisson processes, which they showed to either stochastically dominate or be dominated by the processes of interest.

\paragraph{Relaxed Process $\{\mathcal S_i\}_{i\in \calG}$.}
We similarly introduce a stochastic dominance relation where the arrivals and departure of each good are independent of each other. We define stochastic process $\mathcal S_i$ for each good $i$ where we only sell good $i\in G$ at a higher rate than $s_{ij}$. More formally, we consider $\mathcal S_i$ where items of good $i$ are supplied according to the Poisson process with rate $\lambda_i>0$ and perish at exponential rate $\mu_i$, Buyer of type $j$ arrives at rate $\gamma_j$ and an arrived buyer receives an item of good $i$ (if available) with probability $q_{ij}$. Crucially, the stochastic processes $\{\mathcal S_i\}_{i\in \mathcal G}$ are mutually independent. We denote the set of available goods in the relaxed process by $\tilde A$, i.e. goods with at least one item available in their respective processes. 
 We also let $\tilde R_j$ be the set that includes each available good in the relaxed process $i\in \tilde A$ with probability $q_{ij}$ independently. Next, we show that the probability $\Pr[S\cap R_j \neq  \emptyset]$ can be lower bounded by $\Pr[S\cap \tilde R_j \neq  \emptyset]$. The proof of the lemma is similar to that of \cite[Lemma 3.2 and Claim 3.1]{kessel2022stationary}, and is therefore deferred to \Cref{proof of stochastic dominance}. 

\begin{lem} \label{lem:stochastic dominance}
For any set of goods $S\subseteq \calG$,
\begin{equation*}
    \Pr[S\cap R_j\neq \emptyset] \geq \Pr[S\cap \tilde R_j \neq \emptyset] \geq \left(1 - \frac 1 {\sqrt e} \right)\cdot \sum_{i\in S}\frac{x_{ij}}{\gamma_j}.
\end{equation*}
\end{lem}

Lemma~\ref{lem:stochastic dominance} combined with Theorem~\ref{thm:single_choice_CRS} implies Lemma~\ref{lem:lowerbound_on_no_proposal}. We are now ready to prove Theorem~\ref{thm:multi-good}.

\begin{proof}[Proof of Theorem~\ref{thm:multi-good}]
  Let $s_{ij}$ be the selling rate of \Cref{alg:main_new}. We can lower bound \Cref{alg:main_new}'s performance as:
  \begin{align*}
      \sum_{j \in \calB}\sum_{i\in \calG}v_{ij} \cdot s_{ij} &= \sum_{j \in \calB} \gamma_j \cdot \sum_{i\in \calG} v_{ij} \cdot  \Pr[i\in \pi(R)] && \text{(Lemma~\ref{lem:lowerbound_sellingrate_singlechoice})}\\
      &\geq \sum_{j \in \calB} \gamma_j \cdot \sum_{i\in \calG}  v_{ij} \cdot \left( 1 -\frac{1}{\sqrt e} \right)\cdot \frac{x^*_{ij}}{\gamma_j} && \text{(Lemma~\ref{lem:lowerbound_on_no_proposal})}\\
      &=\left(1 -\frac{1}{\sqrt e}\right) \cdot \sum_{j\in \calB} \sum_{i\in \calG}v_{ij}\cdot x^*_{ij} \\
      & \geq \left(1 -\frac{1}{\sqrt e}\right) \cdot \optoff. && (\textrm{\Cref{cor:RB-geq-OPT}}) \qedhere
  \end{align*}
\end{proof}

\section{Conclusion and Open Questions}

We provided a wide range of results for the stationary prophet inequality problem, an infinite time-horizon counterpart of the classic prophet inequality problem; our main result are asymptotically tight characterizations for a wide range of combinatorial constraints on the buyers' demands, for both linear and submodular valuations.
Our work also suggests a wealth of follow-up questions.

\paragraph{Tight Bounds.}
We provide asymptotically tight characterizations of combinatorial (ex-ante) SPIs in terms of offline CRS. 
However, only for the most basic, single-good SPI problem with unit-demand buyers do we currently know optimal bounds \cite{kessel2022stationary}.
Can the bound of $1/2$ for the simple problem be achieved for the more general \emph{multi-good} problem with unit demands, improving our bound of $1-\frac{1}{\sqrt{e}}$? 
What about matroid demands? Alternatively, are these problem strictly harder than the basic SPI problem?
Generally, for what natural constraints can we design optimal SPIs?

\paragraph{Approximating the Optimum Online.}
We provide polytime algorithms approximating the optimal matroid-SPI algorithm. 
What is the best approximation achievable for this (and other) SPI problems? For which classes can we provide (F)PTASes, and which are APX hard?

\paragraph{Richer valuations classes.} 
Our results extend beyond linear valuations, to submodular valuations. 
Can similar results be obtained for still richer valuation functions, such as subadditive or XOS valuations?

\paragraph{Posted-Price Mechanisms.} The single-good algorithms of \cite{kessel2022stationary} are pricing-based, resulting in incentive-compatible mechanisms for these dynamics, but only for unit-demand buyers in the single-good setting. 
Can similar results be obtained for combinatorially-constrained buyers (with multiple goods), mirroring the rich work on this question for the classic stationary prophet problem?

\paragraph{Acknowledgements.} The authors thank Ziv Scully for pointing out \cite{ross1995stochastic} to their attention.

\appendix
\section{Deferred Proofs: Reductions to Single-Good SPI}\label{appendix:combinatorial}

In this section we prove the following lemmas from sections \ref{sec:combinatorial} and \ref{sec:Opton}, repsectively.
\savedgivenproposed*
\optonreduction*

\begin{claim}\label{claim: A|R lowerbound G_H}
     For $i\in \calG_H$, we have $$\frac{\Pr[i\in S]}{w_i} \geq \min_{x\geq 0} \frac{1-  \left(1 + \sum_{q=1}^\infty \prod_{r=1}^q \frac{x}{r  +  c \cdot x/(1-\exp(-x))} \right)^{-1}}{1 - \exp\left( -x\right)}.$$
\end{claim}

Effectively, our proof technique for both lemmas is to reduce to the analysis of the single-good SPI algorithms of \cite{kessel2022stationary}.
To prove the above, we first need the following standard fact for the kind of birth-death processes we study. (See, e.g., \cite{kessel2022stationary}.)

\begin{restatable}{prop}{PrAvailableExact}
\label{lem:Pr-available-exact}
For any online algorithm that sells any available item of good $i \in \calG$ to buyers which arrive at rate $\gamma^*\geq 0$, the stationary probability of an item of good $i$ being available satisfies
\[  \bbP \left[i\in A\right] = 1 - \left(\sum_{q=0}^\infty \prod_{r=1}^q \frac{\lambda_i}{r \cdot \mu_i + \gamma^*} \right)^{-1} \in \left[1-\left(\sum_{q=0}^\infty \frac{1}{q!} \left( \frac{\lambda_i}{ \mu_i + \gamma^*} \right)^q \right)^{-1},\, 1- \exp\left( -\frac{\lambda_i}{\mu_i} \right) \right]. \]
\end{restatable}

\begin{proof}[Proof of Lemma~\ref{lem:A|R}]
First, since $i\in R_j = i\in P \land (X_{ij}=1)$, for $X_{ij}\sim\Ber(q_{ij})$ independent of $i\in P$, together with Bayes' Law, we have that
$\Pr[i\in S \mid i\in R_j] = \Pr[i\in S \mid i\in P]$. So, we wish to lower bound $\Pr[i\in S \mid i\in P]$.
To this end, we note that the number of present and saved items of good $i$ follows a birth-death process as described in Proposition~\ref{lem:Pr-available-exact}, with the former having sale rate of zero, and the latter having a sale rate of $\gamma^*$ upper bounded as follows.
$$\gamma^* \leq \sum_{j\in \calB}  \gamma_j \cdot q_{ij} \leq \lambda_i/w_i.$$ 
\color{black}
Above, the first inequality follows from proposal being a prerequisite condition for allocation, and the second inequality follows from Constraint~\eqref{eqn:flow-constraint-seller}.
Thus, by Proposition~\ref{lem:Pr-available-exact}, we have the following
\begin{equation*}
   \Pr[i\in S \mid i\in P] = \frac{ \Pr \left[ i\in S\right] }{\Pr[i \in P]}\geq \frac{1-  \left(1 + \sum_{q=1}^\infty \prod_{r=1}^q \frac{\lambda_i}{r \cdot \mu_i +  \lambda_i/(1-\exp(-\lambda_i/\mu_i))} \right)^{-1}}{1 - \exp\left( -\lambda_i/\mu_i\right)}\geq \frac{1}{2},
\end{equation*}
where the last inequality follows from  
\cite[Claim~4.3 in arXiv version]{kessel2022stationary} applied to $x=\lambda_i/\mu_i$.
\end{proof}

\begin{proof}[Proof of Lemma~\ref{lem:A|R G_L}]
    The number of  saved items of good $i$ follows a birth-death process as described in Proposition \ref{lem:Pr-available-exact}, with the former having sale rate of zero, and the latter having a sale rate of $\gamma^*$ upper bounded as follows.
$$\gamma^* \leq \sum_{j\in \calB}  \gamma_j \cdot q_{ij} \leq \lambda_i/w_i.$$  
Therefore, by Proposition~\ref{lem:Pr-available-exact}, we have the following
\begin{equation*}
    \frac{ \Pr \left[ i\in S\right] }{w_i}\geq \frac{1-  \left(1 + \sum_{q=1}^\infty \prod_{r=1}^q \frac{\lambda_i}{r \cdot \mu_i +  \lambda_i/w_i} \right)^{-1}}{w_i}.
\end{equation*}

    The proof follows from  \cite[Claims~4.4 and 4.5 in arXiv version]{kessel2022stationary} applied to the above bound by letting $x=\lambda_i/\mu_i$ and $w =\min \left( \Pr[i\in P], \frac{\lambda_i - \sum_{\ell \in \calB} x_{i\ell}}{\mu_i} \right)$.
\end{proof}

\section{Deferred Proofs of \Cref{sec:Opton}}\label{appendix:OPTon}

\extendedCRS*
\begin{proof}
We denote the rank function of matroid $\calM$ by $r_\calM()$. Consider any non-negative weight vector $\vec y$ of elements of matroid $\calM$, for $R\sim \calG(\vec x)$, and concave closure $r^+_\calM$, we prove the following inequality:
\begin{equation*}
 \frac{\mathbb E \left[ r_{\calM} (R) \right]}{r^+_\calM (\vec x)}  = \frac{\mathbb E \left[\max_{S \subseteq R(\vec x), S\in \calF}  \sum_{i\in S} y_i \right]}{r^+_\calM (\vec x)}   \geq \frac 1 c \cdot (1-e^{-c}).
\end{equation*}     
The above equation combined with the characterization of CRSes in \cite{calinescu2011maximizing,dughmi2019outer} shows the existence of a CRS map that rounds the set $R \sim \calG(\vec x)$ to a feasible set in matroid $\calM$ such that each element in $R$ is selected by CRS with probability $\geq \frac 1 c (1-\exp(-c))$.

To prove the above inequality, for each $e\in \calG$, we set up an independent Poisson clock $\mathcal C_e$ of rate $x_e ' = \frac {x_e}{c}$. We define a random process that starts with an empty set $S(0) = \emptyset$ at time $t = 0$. At any time when the clock $\mathcal C_e$ sends a signal, we include element $e$ in $S$, which increases its value by $f(e\mid S)$. (If $e$ is already in S, the marginal value $f(e\mid S)$ is zero.) 
Let $S(t)$ be the random set selected during the process until time $t$. By the definition of a Poisson clock, $S(c)$ contains element $e$ independently with probability $1-e^{-c \cdot \frac{x_e}{c}}\leq x_e$.
We have $\E[r_\calM(S(c))]\leq \E_{R\sim \calG(\vec x)} [ r_\calM(R)]$ due to monotonicity of the rank function of matroid. Now, we compute the change in $r_\calM(S(t))$ in infinitesimal interval $[t , t+dt]$. The probability that $e\in \calG$ is added to $S$ is $x_e' dt$. Since $dt \rightarrow 0$, we can focus on the case when only one of the clocks sends a signal. Thus the expected increase of $r_{\calM}(S(t))$ in the interval $[t,t+dt]$ (ignoring $O(dt^2)$ terms)
\begin{align*}
  \E [r_{\calM}(S(t+dt)) - r_{\calM}(S(t))\mid S(t) = S] &\geq \left( \sum_{e\in E} r_\calM(e\mid S)\cdot x_e' \right)dt\\
  &\geq \left(\min_{S\subseteq \calG} \left \{ r_\calM(S) - \sum_{e\in E} r_\calM(e\mid S)\cdot x_e' \right\} - r_\calM(S)\right) dt.
\end{align*}
 Above, the last inequality holds by adding and subtracting $r_\calM(S)\cdot dt$. Now recall the definition of concave closure of $r^+_\calM(\vec x')$ from Definition~\ref{def:concave_closure}. For any feasible vector $\alpha_S$ for the vector $\vec x'$ and any $T\subseteq \calG$, we have
 \begin{equation*}
     \sum_{S\subseteq \calG} \alpha_S \cdot r_\calM(S) \leq \sum_{S\subseteq \calG} \alpha_S \cdot \left( r_\calM (T) + \sum_{e\in \calG} f(e\mid T)\right) = r_\calM(T) - \sum_{e\in E} r_\calM(e\mid T)\cdot x_e'.
 \end{equation*}
Above the first inequality holds due to the submodularity of $r_\calM$. This implies that,
\begin{align*}
    \min_{S\subseteq \calG} \left \{ r_\calM(S) - \sum_{e\in E} r_\calM(e\mid S)\cdot x_e' \right\} \geq r_\calM^+(\vec x' ).
\end{align*}
This further implies that,
\begin{align*}
  &\E [r_{\calM}(S(t+dt)) - r_{\calM}(S(t))\mid S(t) = S] \geq (r^+_\calM(\vec x') - r_\calM(S))dt\\
  \implies & \frac{d}{dt}\E[r_\calM(S(t))] \geq r^+_\calM(\vec x') - \E[r_\calM(S(t))]\\
  \implies & \E[r_\calM(S(c))] \geq   (1 - \exp(-c)) \cdot r_\calM(\vec x') \\
  \implies & \E_{R\sim \calG(\vec x)}[r_\calM(R)] \geq \frac{1}{c} \cdot (1 - \exp(-c)) \cdot r_\calM(\vec x). 
\end{align*}
Above, the first step follows by taking expectations over $S(t)$. The second step follows by solving the differential equation from $t = 0$ to $t=c$. The last inequality holds because $\E[r_\calM(S(c))] \leq \E_{R\sim \calG(\vec x)}[r_\calM(R)]$ and $r_\calM^+$ is concave which implies $r_\calM(\vec x') \geq \frac{1}{c} \cdot r_\calM(\vec x)$. Combining everything, we conclude the proof. 
\end{proof}

\availablegh*
\begin{proof}
    For $i\in \calG_H$, similar to  Lemma~\ref{lem:A|R G_L} and Lemma~\ref{lem:A|R}, 
    we obtain 
    \begin{equation*}
        \frac{\Pr[i\in S]}{w_i} = c \cdot \frac{\Pr[i \in S]}{\Pr[i\in P]} \geq  c \cdot \frac{1-  \left(1 + \sum_{q=1}^\infty \prod_{r=1}^q \frac{\lambda_i}{r \mu_i  +  c \cdot \lambda_i/(1-\exp(-\lambda_i / \mu_i))} \right)^{-1}}{1 - \exp\left( -\lambda_i / \mu_i\right)}.
    \end{equation*}
    Above, the first equality follows by the definition of $w_i$ for goods $i\in \cal G_H$ and the last equality follows due to Proposition~\ref{lem:Pr-available-exact}. Next, by setting $z = \frac{\lambda_i}{\mu_i}$ we obtain an analytical bound on the ratio $\frac{\Pr[i\in S]}{w_i} \geq c \cdot \max_{x\geq 0} h_c(x)$, where $$h_c (x): = \frac{1-  \left(1 + \sum_{q=1}^\infty \prod_{r=1}^q \frac{x}{r  +  c \cdot x/(1-\exp(-x))} \right)^{-1}}{1 - \exp\left( -x\right)}.$$ To complete the proof, next we show that the function $h_c (x)$ is lower bounded by $\frac{1}{1+c}$. As a first step, obtain a manageable lower bound on $h_c(x)$ as follows:
    \begin{align*}
        h_c(x) &=  \frac{1-  \left(1 + \sum_{q=1}^\infty \prod_{r=1}^q \frac{x}{r  +  c \cdot x/(1-\exp(-x))} \right)^{-1}}{1 - \exp\left( -x\right)}\\
        &\geq \frac{1-  \left(1 + \sum_{q=1}^\infty \prod_{r=1}^q \frac{x}{r \left(1 +  c \cdot x/(1-\exp(-x)) \right)} \right)^{-1}}{1 - \exp\left( -x\right)}\\
        &=\frac{1 - \exp \left( - \frac{x}{1 +  c \cdot x/(1-\exp(-x))}\right)}{1 - \exp\left( -x\right)}. 
    \end{align*}
Above, the inequality holds because $r>1$, the last equality holds due to Taylor's expansion of $e^x$. For sake of exposition, we let $\alpha_c(x):=\frac{1 - \exp \left( - \frac{x}{1 +  c \cdot x/(1-\exp(-x)) }\right)}{1 - \exp\left( -x\right)}$. We first observe that $\lim_{x\rightarrow 0+} \alpha_c(x) = \frac 1 {1+c}$. Next, we show that for any $c>0$, the function $\alpha_c(x)$ is increasing in $x$ which will conclude the proof of the claim. We differentiate the function $\alpha_c(x)$ and simplify to obtain:
\begin{align*}
    \alpha'_c(x) = \frac{e^{x \left(2 - \frac{e^x - 1}{e^x (c x + 1) - 1}\right)} (-c x^2 + e^x (c\cdot x^2 (c + 1)  + 2 c x - 2) - 2 c x + e^{2 x} + 1)}{(e^x - 1)^2 (e^x (c x + 1) - 1)^2}.
\end{align*}
We can observe that the above  denominator is non-negative, namely that $\frac{e^{x \left(2 - \frac{e^x - 1}{e^x (c x + 1) - 1}\right)}}{(e^x - 1)^2 (e^x (c x + 1) - 1)^2} \geq 0$. Hence we only analyze the numerator in the preceding equation and show that it too is non-negative. 
\begin{align*}
    -c x^2 + e^x (c\cdot x^2 (c + 1)  + 2 c x - 2) - 2 c x + e^{2 x} + 1)
    &= -c x^2 + c x^2 e^x + c^2 x^2 e^x + 2c x e^x -2e^x +e^{2x}+1\\
    &= c x^2 \cdot (e^x -1)+ 2c x (e^x  -1) + e^{2x} - 2e^x - 1\\
    & = c x^2 \cdot (e^x -1)+ 2c x (e^x  -1) + (e^{x}  - 1)^2\\
    &\geq 0.\qedhere
\end{align*}
\end{proof}


\section{Deferred Proofs of \Cref{sec:multi-good}: Mutli-Good SPIs}\label{appendix:multi-good}

In our analysis of the multi-good problem we will need to prove stochastic dominance between two processes, for which the following lemma will prove useful. For this lemma, we recall that a set $S\subseteq \calY \subseteq \mathbb{R}^n$ is or upward closed if for every $y\geq \tilde y$ with $\tilde y \in S$ and $y\in \calY$, we have that $y \in S$. 

\begin{lem}[\cite{brandt1994pathwise}]
\label{lem:stochastic-dominance}
Let $Y, \tilde Y$ be two stochastic processes taking values in $\calY\subseteq \mathbb{R}^n$, with time-homogeneous intensity matrices $Q, \tilde Q$.
Then, $Y$ stochastically dominates $\tilde Y$ ($\Pr[Y\geq y]\geq \Pr[ \tilde Y\geq y]$ for all $y\in \calY$) if and only if the following holds: for every $y, \tilde y \in \calY$ and upward closed set $S\subseteq \calY$, if $y\geq \tilde y$, and either $y, \tilde y \in S$ or $y, \tilde y \notin S$, then
\[ \sum_{z \in S} Q(y, z) \geq \sum_{z \in S} \tilde Q (\tilde y, z).\]
\end{lem}

\subsection{Proof of Lemma~\ref{lem:stochastic dominance}}\label{proof of stochastic dominance}

We let $Y \in \R^{2n}$ be the vector whose elements, which we refer to as $Y_{A_i}$ and $Y_{P_i}$, represent the number of items of good $i$ available and the negative of the number of items of good $i$ present, respectively, under \Cref{alg:main_new}. We use $e_{A_i}$ and $e_{P_i}$ to denote the vectors with all zeros except at the elements corresponding to $A_i$ and $P_{i}$ which are $1$.
We can think of $Y$ as simply the (augmented) state of the marketplace under \Cref{alg:main_new}, where the set $\calY \subseteq \mathbb{R}^{2n}$ of valid states is such that for any $y \in \calY$, we have $0 \leq y_{A_i} \leq C$, $y_{P_i} \leq 0$, and $y_{A_i} \leq |y_{P_i}|$.
Under \Cref{alg:main_new}, the stochastic process governing $Y$ is described by intensity matrix $Q$, where for any $y, y^\prime \in \calY$,
\begin{equation}\label{eq:stoc-process-alg1}
Q(y, y^\prime) = \begin{cases}
	\lambda_i &  y^\prime = y + e_{A_i} - e_{P_i} \\
	y_{A_i} \cdot \mu_i &  y^\prime = y - e_{A_i} + e_{P_i} \\
		\left(|y_{P_i}| - y_{A_i} \right) \cdot \mu_i &  y^\prime = y + e_{P_i} \\
	    \sum_{j \in \calB} \gamma_j \cdot \Pr[i\in \pi(R) \mid i\in A] &  y^\prime = y - e_{A_i} \text{ and } y_{A_i} > 0 \\
		0 & \text{o.w.}
	\end{cases}
	\end{equation}

	and $Q(y, y) = - \sum_{y^\prime \in \calY : y^\prime \neq y} Q(y, y^\prime)$.
	
	Although the availability of good $i$ and the presence of other goods $i^\prime \neq i$ are correlated under $Q$, we show that $Y$ stochastically dominates a stochastic process $\tilde Y$ under which they are, in fact, independent across elements. $\tilde Y$ can be thought of as a collection of $n$ independent single-good instances, where each instance consists of a different good $i \in \calG$ and the full set of buyers $\calB$. 
	
	More specifically, we let $\tilde Y$ represent the state of a stochastic process on the same space $\calY$ governed by intensity matrix $\tilde Q$, where for any $y, y^\prime \in \calY$,
	\begin{equation}\label{eq:ind-process}
	\tilde Q(y, y^\prime) = \begin{cases}
		\lambda_i &  y^\prime = y + e_{A_i} - e_{P_i}  \\
		y_{A_i} \cdot \mu_i &  y^\prime = y - e_{A_i} + e_{P_i} \\
		\left(|y_{P_i}| - y_{A_i} \right) \cdot \mu_i &  y^\prime = y + e_{P_i} \\
		\sum_{j \in \calB} \gamma_j \cdot \frac{b\cdot x_{ij}}{\gamma_j\cdot w_i}&  y^\prime = y - e_{A_i} \text{ and } y_{A_i} > 0 \\
		0 & \text{o.w.}
	\end{cases}
	\end{equation}
	and $\tilde Q(y, y) = - \sum_{y^\prime \in \calY : y^\prime \neq y} \tilde Q(y, y^\prime)$. Observe that the $Q$ and $\tilde Q$ are identical except at the rate at which they make a transition to the state with one fewer available item.  
 
 The execution of \Cref{alg:main_new} when a buyer of type $j$ arrives is equivalent to the following: the seller first determines which goods can be sold to buyer $j$ by sampling a set of permissible goods from the product distribution $\Ber(q_{1j}) \times \dots \times \Ber(q_{nj})$, denoted as a set $H_j$.
 Next, Following the proof of Lemma 3.2 and Claim A.1 from \cite{kessel2022stationary} we conclude that the stochastic process $Y \mid \{H_j\}_{j\in \calB}$ stochastically dominates the stochastic process $\tilde Y\mid \{H_j\}_{j\in \calB}$. This implies that for any $S\subseteq \calG$,
 \begin{align*}
    \Pr[S\cap R_j \neq \emptyset ] &= 1-\Pr[S\cap R_j = \emptyset]\\ 
    &=1 - \sum_{H_j\subseteq \calG} \Pr\left[Y \leq \sum_{i\in H_j\cap S} e_{A_i} \;\middle\vert\; \{H_j\}_{j\in \calB}\right]\\
    &\geq 1 - \sum_{H_j\subseteq \calG} \Pr\left[\tilde Y \leq \sum_{i\in H_j} e_{A_i} \;\middle\vert\; \{H_j\}_{j\in \calB}\right] && (Y \mid \{H_j\}_{j\in \calB} \succeq \tilde Y \mid \{H_j\}_{j\in \calB})\\
    & = 1- \Pr[ S \cap \tilde R_j = \emptyset]\\
    & = \Pr[S\cap \tilde R_j \neq \emptyset].
 \end{align*}
 Above, the second  and the second last equality holds because if all goods in $i\in H_j\cap S$ are not available then $S\cap R = \emptyset$. To complete the proof, we consider
 \begin{align*}
\Pr[S\cap \tilde R_j \neq \emptyset] &= 1- \Pr[ S \cap \tilde R_j = \emptyset]= 1 - \prod_{i\in S} \Pr[i\notin \tilde R_j]= 1 - \prod_{i\in S} \left(1 - \frac{x_{ij}}{\gamma_j \cdot w_i}\Pr[i\in \tilde A]\right)\\
&\geq 1 - \prod_{i\in S} \left(1 - \frac{x_{ij}}{2 \cdot \gamma_j }\right) \geq \left(1 - \frac 1 {\sqrt e} \right) \sum_{i\in S} \frac{x_{ij}}{\gamma_j}. 
 \end{align*}
 Above, the second inequality holds because the events $\{i\in \tilde R_j\}_{i\in S}$ are independent, as $\{\mathcal S_i\}_{i\in S}$ are independent. The first inequality holds because $\frac{\Pr[i\in \tilde A]}{w_i} \geq \frac 1 2$ (follows from Lemma~\ref{lem:A|R}). The last inequality holds because $1-\exp(-bx)\geq (1-\exp(-b))x$ by convexity. 
\section{Challenging Correlation for \Cref{alg:combinatorial} }\label{appendix:challenging_correlation}
In this section we construct an instance where the events $\{i\in \pi(R)\}$ and $\{i\in A\}$ are (strictly) negatively correlated due to the positive correlations between the availability of good $i$ and presence of some other good $i'$. To make this argument concrete, consider the following example:
\begin{example}\label{example:Bad correlation for alg}

    Consider an SPI instance with one-uniform matroid constraint with two goods $\calG = \{a,b\}$ and one buyer $\calB = \{c\}$. Items of good $a$ are produced at rate $1$ and perish at rate $\frac{1}{\log \frac 1 \epsilon}$, and items of good $b$ are produced at rate $\delta \cdot \epsilon$ and perish at rate $\epsilon$. The valuation for the buyers are $v_{ac} = v_a = 1+\epsilon$ and $v_{bc}=v_b=1$. The buyer $c$ arrives at rate $\gamma_c = 1$. We can observe that the optimal selling rate $\vec x^*\in \arg\max\{\vec v\cdot \vec x \mid \vec x\in \calP_{off}(\calI)\}$ is $x^*_a = 1 - \exp(-\log 1/\epsilon) = 1 - \epsilon$ and $x^*_b = \epsilon$. We consider the CRS scheme as follows:  $\pi(\{a, b \}) = a, \pi(\{a\}) = a, \pi(\{b\})= b, \pi(\{\emptyset \}) = \emptyset $.  We now compute $q_a$ and $q_b$ as follows: $$q_a = \frac{x_a^*}{\Pr[a\in P]} = 1 \text{ and } q_b = \frac{x_b^*}{\Pr[b\in P]}.$$ In addition, the probability of the proposals can be expressed as: $\Pr[a \in R] = x_a^* \text{ and } \Pr[b \in R] = x_b^*.$ We also have $$\Pr[a\in P]  = 1 - \epsilon \text{ and } \Pr[b\in P] = 1 - \exp(-1/\delta).$$
\end{example}
\begin{lem}\label{lem:badcorrlation}
    On \Cref{example:Bad correlation for alg}, for $\delta = 0.3$, we have $\Pr[b\in \pi(R_c)\cap A] , \Pr[b\in \pi(R_c)], \Pr[b \in A] >0$. Moreover, $$\Pr[b\in \pi(R_c)\cap A] \leq 0.81 \cdot \Pr[b\in \pi(R_c)] \cdot \Pr[b \in A].$$ 
\end{lem}
The intuition here is that when good $a$ is not present then an item or good $b$ are sold to the buyer with the positive rate in contrast to when the items of good $a$ are present in which case an item of good $b$ is never sold. Hence, intuitively, $\Pr[b\in A]$ should be larger than $\Pr[b\in A \mid a \notin P]$. This precisely positively correlates the availability of good $b$ with the presence of good $a$ which leads unfavorable correlation for \Cref{alg:combinatorial}.

To formally prove Lemma~\ref{lem:badcorrlation}, we first upper bound the probability $\Pr[b\in \pi(R_c)\cap A]$ and lower bound the product of the probabilities  $\Pr[b\in \pi(R_c)] \cdot \Pr[b \in A]$. 

\begin{claim}
   On \Cref{example:Bad correlation for alg}, $\Pr[b\in \pi(R)\cap A] =  \Pr[b\in A \mid a\notin P ] \cdot \Pr[a \notin P] \cdot  q_b$.
\end{claim}
\begin{proof}
    We first notice that when an item of good $a$ is present then it proposes to the buyer with probability $1$. In addition, CRS $\pi$ outputs good $a$ whenever it proposes. Therefore, $b\in \pi(R)$ iff $a\notin P$ and $b\in R$. Hence, 
    \begin{equation*}
        \Pr[b\in \pi(R)\cap A] =  \Pr[b\in A\land a \notin P \land b\in R] = \Pr[b\in A \cap a\notin P  ] \cdot q_b = \Pr[b\in A \mid a\notin P ] \cdot \Pr[a \notin P] \cdot  q_b. 
    \end{equation*}
\end{proof}

Next, we compute the product product of probability $\Pr[b\in \pi(R_c)] \cdot \Pr[b \in A]$.
\begin{claim}
  On \Cref{example:Bad correlation for alg}, $$\Pr[b\in \pi(R)] \cdot \Pr[b\in A] = q_b \cdot  \Pr[ a \notin P] \cdot \Pr[b\in A] \cdot \left( 1 - \exp \left( - \frac{1}{\delta} \right) \right).$$  
\end{claim}
\begin{proof}
We have,
\begin{align*}
    \Pr[b\in \pi(R)] \cdot \Pr[b\in A] &= q_b \cdot \Pr[b\in R \land a \notin P] \cdot \Pr[b\in A]\\
    &= q_b \cdot  \Pr[ a \notin P] \cdot \Pr[b\in A] \cdot \left( 1 - \exp \left( - \frac{1}{\delta} \right) \right)
\end{align*}
Above, the last inequality holds because events $\{a\notin P\}$ and $\{b \in R\}$ are independent. Next, we lower bound $\Pr[b\in A]$ and upper bound $\Pr[b\in A \mid a\notin P]$. 
\end{proof}
Next, we bound the probabilities $\Pr[b\in A]$ and $\Pr[b\in A\mid a \notin P]$ to complete the proof of Lemma~\ref{lem:badcorrlation}.

\begin{claim}\label{claim:bA|anoP}
On \Cref{example:Bad correlation for alg},
    \begin{equation*}
        \Pr[b \in A \mid a\notin P] = 1 - \left( \sum_{q=0}^\infty\prod_{r=0}^q \frac{\epsilon}{r\cdot \delta \cdot \epsilon + q_b} \right)^{-1} \in \left(1 -\exp\left( - \frac{1}{q_b/\epsilon+\delta} \right) , \frac{1}{\delta+q_b/\epsilon}\right).
    \end{equation*}
\end{claim}
\begin{proof}
    We note that the number of available items of good $b$ conditioned on $a\notin P$ again follows a birth-death process with the former having sale rate of zero, and the latter having a sale rate of $\gamma^* = \gamma_c \cdot q_b  = \epsilon $. Here, we note that the an item of available good $b$ is sold whenever it proposes to the buyer as we conditioned on $a\notin P$. Hence, we can upper bound
    \begin{align*}
        \Pr[b\in A \mid a\notin P] &= 1 - \left( \sum_{q=0}^\infty\prod_{r=0}^q \frac{\epsilon}{r\cdot \delta \cdot \epsilon + q_{b}} \right)^{-1}\\
        &\leq 1 - \left( \sum_{q=0}^\infty\prod_{r=0}^q \frac{\epsilon}{ \delta \cdot \epsilon + 
         q_{b}} \right)^{-1}\\
         &=  \frac{\epsilon}{\delta\epsilon + q_b} && \text{(Sum of geometric progression)}\\
        & =\frac{1}{\delta+q_b/\epsilon}.
    \end{align*}
Next, we obtain a lower bound as follows:
\begin{align*}
    \Pr[b\in A \mid a\notin P] &= 1 - \left( \sum_{q=0}^\infty\prod_{r=0}^q \frac{\epsilon}{r\cdot \delta \cdot \epsilon + q_{b}} \right)^{-1} \\& \geq  1 - \left( \sum_{q=0}^\infty\prod_{r=0}^q \frac{\epsilon}{ r\cdot( \delta \cdot \epsilon + 
         q_{b})} \right)^{-1}\\
    &= 1 - \left( \sum_{q=0}^\infty\frac{1}{q!} \cdot \left(\frac{\epsilon}{  \delta \cdot \epsilon + 
         q_{b}} \right)^q \right)^{-1}\\
    & = \exp\left( - \frac{1}{q_b/\epsilon+\delta} \right) \qquad \qquad \qquad \text{(Taylor Expansion)} \qedhere 
\end{align*}
\end{proof}

\begin{claim}\label{claim:bA|aP}
On \Cref{example:Bad correlation for alg},
    \begin{equation*}
        \Pr[b \in A]  \geq \left(1 -\exp\left( - \frac{1}{\frac{1}{1-\exp(-1/\delta)}+\delta} \right) \right) \cdot \epsilon +  \left( 1- \exp \left( - \frac{1}{\delta}\right) \right) \cdot  \left( 1- \epsilon \right).
    \end{equation*}
\end{claim}
\begin{proof}
    We note that the number of available items of good $b$ follows a birth-death process as described in Proposition~\ref{lem:Pr-available-exact}, with the former having sale rate of zero, and the latter having a sale rate of $\gamma^* = \gamma_c \cdot q_b \cdot \Pr[a\notin P \mid a \in P] = 0$  Here, we note that the an item of available good $b$ is sold iff an item of good $a$ is not present. Once we condition on the event $a\in P$ then an item of $b$ never gets sold. This implies that:
    \begin{align*}
        \Pr[b\in A\mid a\in P] & =  1 - \left( \sum_{q=0}^\infty\prod_{r=0}^q \frac{\epsilon}{r\cdot \delta \cdot \epsilon} \right)^{-1}  =1 - \left( \sum_{q=0}^\infty \frac{1}{q!} \cdot \left(\frac{1}{\delta} \right)^q \right)^{-1}\\
        &= 1- \exp \left( - \frac{1}{\delta}\right) \qquad \qquad \qquad \text{(Taylor Expansion)} 
    \end{align*}
Claim~\ref{claim:bA|anoP} with above inequality implies that
\begin{align*}
    \Pr[b\in A] &= \Pr[b\in A \mid a\notin P]\cdot \Pr[a\notin P] + \Pr[b\in A \mid a\in P]\cdot \Pr[a\in P]\\
    &\geq \left(1 -\exp\left( - \frac{1}{q_b/\epsilon+\delta} \right) \right) \cdot \epsilon +  \left( 1- \exp \left( - \frac{1}{\delta}\right) \right) \cdot  \left( 1- \epsilon \right)\\
    &=\left(1 -\exp\left( - \frac{1}{\frac{1}{1-\exp(-1/\delta)}+\delta} \right) \right) \cdot \epsilon +  \left( 1- \exp \left( - \frac{1}{\delta}\right) \right) \cdot  \left( 1- \epsilon \right).
\end{align*}
Above, the inequality holds because of Claim~\ref{claim:bA|anoP}, Claim~\ref{claim:bA|aP} and $\Pr[a\in P] = 1 - \epsilon$. The second equality holds because $q_b = \frac{\epsilon}{\Pr[b\notin P]} = \frac{\epsilon}{1 - \exp(-1/\delta)}$.
\end{proof}

Now, we are ready to prove the main lemma.
\begin{proof}[Proof of Lemma~\ref{lem:badcorrlation}]
 Using Claim~\ref{claim:bA|anoP} and Claim~\ref{claim:bA|aP}, we have all the required probabilities $\Pr[b\in \pi(R_c)\cap A], \Pr[b\in \pi(R_c)], \Pr[b \in A]$ are positive. Moreover,
\begin{align*}
    &\frac{\Pr[i\in A\cap \pi(R)] }{ \Pr[i\in A] \cdot \Pr[i\in \pi(R)]}  =\frac{q_b \cdot \Pr[b\in A \mid a\notin P]\cdot \Pr[a \notin P]}{q_e \cdot \Pr[b\in A] \cdot \Pr[a\notin P] \cdot (1- \exp(-1/\delta)) }\\
    &=\frac{1}{1-\exp(-1/\delta)} \cdot \frac{\Pr[b\in A\mid a\notin P]}{\Pr[b\in A]}\\
    &\leq \frac{\frac{1}{\delta+1/(1-\exp(-1/\delta))}}{ (1-\exp(-1/\delta)) \cdot \left(\left(1 -\exp\left( - \frac{1}{\frac{1}{1-\exp(-1/\delta)}+\delta} \right) \right) \cdot \epsilon +  \left( 1- \exp \left( - \frac{1}{\delta}\right) \right) \cdot  \left( 1- \epsilon \right)\right)}.
\end{align*}
Now taking $\epsilon \rightarrow 0+$, we get,
\begin{equation*}
    \frac{\Pr[i\in A\cap \pi(R)] }{ \Pr[i\in A] \cdot \Pr[i\in \pi(R)]} \leq \frac{\frac{1}{\delta+1/(1-\exp(-1/\delta))}}{(1-\exp(-1/\delta)) \cdot \left( \left( 1- \exp \left( - \frac{1}{\delta}\right) \right) \right)}. 
\end{equation*}
Now by setting $\delta = 0.3$, we obtain the result. 
\end{proof}

\bibliographystyle{acmsmall}
\bibliography{abb,ultimate}

\end{document}